\documentclass[a4paper,UKenglish,cleveref, autoref, thm-restate]{lipics-v2021}
\nolinenumbers

\title{Finding equilibria: simpler for pessimists, simplest for optimists} 
\usepackage{amsmath}
\usepackage{amssymb}
\usepackage{amsthm}
\usepackage{algorithm}
\usepackage{algorithmicx}
\usepackage[noend]{algpseudocode}
\usepackage[pdftex,dvipsnames]{xcolor}
\usepackage{tikzlings}
\usepackage{orcidlink}
\usepackage{tikzducks}
\usepackage{tikzsymbols}
\usepackage{tikz}
\usepackage{mathtools}
\usepackage{thmtools}
\usepackage{thm-restate}
\usepackage{hyperref}
\usepackage[capitalise]{cleveref}
\usepackage{xspace}
\usepackage[new]{old-arrows}
\usepackage{float}
\usepackage{wrapfig}
\usepackage{complexity}
\usepackage{caption}
\usepackage{subcaption}
\usepackage{graphicx}
\usepackage{mathabx}
\usepackage{stmaryrd}
\usepackage{wasysym}
\usepackage{pgfplots}
\pgfplotsset{compat=1.17}

\newtheorem*{question}{Question}

\algdef{SE}[DOWHILE]{Do}{doWhile}{\algorithmicdo}[1]{\algorithmicwhile\ #1}%

\newcommand{\tpl}[1]{\left( #1 \right)}

\newcommand{\Gc}{\mathcal{G}}

\newcommand{\Oh}{\mathcal{O}}

\newcommand{\prob}{\mathbb{P}}

\newcommand{\Nb}{\mathbb{N}}
\newcommand{\Rb}{\mathbb{R}}
\newcommand{\Eb}{\mathbb{E}}
\newcommand{\Qb}{\mathbb{Q}}
\renewcommand{\M}{\mathbb{M}}
\newcommand{\re}{\M}
\newcommand{\RM}{\M}
\newcommand{\OM}{\mathbb{OM}}
\newcommand{\PM}{\mathbb{PM}}
\newcommand{\oexp}{\OM}
\newcommand{\pexp}{\PM}
\newcommand{\X}{\mathbb{X}}
\newcommand{\xr}{\X}

\newclass{\TOWER}{TOWER}
\newclass{\ACKERMANN}{ACKERMANN}
\newclass{\EXPTIME}{EXPTIME}
\newclass{\IIEXPTIME}{2-EXPTIME}
\newclass{\NEXPTIME}{NEXPTIME}
\newclass{\SQRTSUM}{SqrtSum}
\newclass{\THREESAT}{3SAT}
\newclass{\PTIME}{PTIME}
\newclass{\RGame}{2PlayerReach}

\usepackage{graphicx}
\usepackage{array}
\usepackage{tikz}
\usetikzlibrary{automata, arrows, positioning, decorations.markings, decorations.pathreplacing}
\usetikzlibrary {decorations.pathmorphing, decorations.shapes}
\usetikzlibrary{backgrounds,fit,shapes.geometric}
\usetikzlibrary {graphs, quotes}
\usetikzlibrary{calc}
\usetikzlibrary{shapes,shapes.geometric,fit}
\usepackage[new]{old-arrows}
\usetikzlibrary {graphs,shapes.geometric,quotes}
\usetikzlibrary{decorations.pathreplacing}
\usetikzlibrary{automata}
\usetikzlibrary{calc}
\usetikzlibrary{arrows,automata,decorations.pathmorphing}
\usetikzlibrary{shapes,shapes.geometric,fit,calc,automata}
\usepackage{xargs} 
\usepackage{todonotes}
\newcommandx{\theju}[2][1=]{\todo[linecolor=blue,backgroundcolor=blue!25,bordercolor=blue,#1]{\tiny T: #2}}
\newcommandx{\leon}[2][1=]{\todo[linecolor=red,backgroundcolor=red!25,bordercolor=red,#1]{\tiny L:#2}}

\usepackage{latexsym}
\usepackage{wasysym}

\theoremstyle{claimstyle}
\newtheorem{invariant}{Invariant}
\newcommand{\modifiedreward}[2]{f_{#1,#2}}
\renewcommand{\Game}{\mathcal{G}}

\DeclareMathOperator{\essinf}{ess\,inf}
\DeclareMathOperator{\esssup}{ess\,sup}
\renewcommand{\|}{\upharpoonright}

\newcommand{\bsigma}{{\bar{\sigma}}}
\newcommand{\btau}{{\bar{\tau}}}

\newcommand{\brho}{{\bar{\rho}}}
\newcommand{\bx}{{\bar{x}}}
\newcommand{\by}{{\bar{y}}}
\newcommand{\bz}{\bar{z}}
\newcommand{\bM}{\bar{M}}

\renewcommand{\d}{\mathsf{d}}
\renewcommand{\p}{\mathsf{p}}
\newcommand{\val}{\mathsf{val}}
\newcommand{\Occ}{\mathsf{Occ}}
\newcommand{\Inf}{\mathsf{Inf}}
\newcommand{\Hist}{\mathsf{Hist}}
\newcommand{\Plays}{\mathsf{Plays}}

\newcommand{\punish}{\mathsf{punish}}
\newcommand{\Supp}{\mathsf{Supp}}
\newcommand{\anchor}{\mathsf{anchor}}
\newcommand{\last}{\mathsf{last}}

\newcommand{\Attr}{\mathsf{Attr}}
\newcommand{\Yes}{\mathsf{Yes}}
\newcommand{\No}{\mathsf{No}}

\newcommand{\Strat}{\mathsf{Strat}}

\newcommand{\bad}{{\frownie}}

\renewcommand{\epsilon}{\varepsilon}
\renewcommand{\l}{\ell}

\newcommand{\stack}[2]{\stackrel{#1}{#2}}

\newcommand{\anch}{\scalebox{0.5}{\text{\faAnchor}}}

\tikzset{stoch/.style={state, fill=black, inner sep=1.5pt, text=white, minimum size=0pt}}

\author{Léonard Brice}{Université Libre de Bruxelles, Belgium}{leonard.brice@ulb.be}{https://orcid.org/0000-0001-7748-7716}{}

\author{Thomas A. Henzinger}{Institute of Science \& Technology Austria}{tah@ist.ac.at}{https://orcid.org/0000-0002-2985-7724}{}

\author{K. S. Thejaswini}{Institute of Science \& Technology Austria}{thejaswini.k.s@ista.ac.at}{https://orcid.org/0000-0001-6077-7514}{}

\authorrunning{L. Brice, T. A. Henzinger, K. S. Thejaswini} 

\Copyright{Léonard Brice, Thomas A. Henzinger, K. S. Thejaswini} 



\category{} 

\relatedversion{} 

\funding{This work is a part of project VAMOS that has received funding from the European Research Council (ERC), grant agreement No 101020093}

\begin{document}

\maketitle

\begin{abstract}
We consider simple stochastic games with terminal-node rewards and multiple players, who have differing perceptions of risk. Specifically, we study risk-sensitive equilibria (RSEs), where no player can improve their perceived reward---based on their risk parameter---by deviating from their strategy. We start with the entropic risk (ER) measure, which is widely studied in finance. ER characterises the players on a quantitative spectrum, with positive risk parameters representing optimists and negative parameters representing pessimists. Building on known results for Nash equilibira, we show that  RSEs exist under ER for all games with non-negative terminal rewards. However, using similar techniques, we also show that the corresponding \emph{constrained} existence problem---to determine whether an RSE exists under ER with the payoffs in given intervals---is undecidable.

To address this, we introduce a new, qualitative risk measure---called \emph{extreme risk} (XR)---which coincides with the limit cases of positively infinite and negatively infinite ER parameters. Under XR, every player is an extremist: an extreme optimist perceives their reward as the maximum payoff that can be achieved with positive probability, while an extreme pessimist expects the minimum payoff achievable with positive probability. Our first main result proves the existence of RSEs also under XR for non-negative terminal rewards. Our second main result shows that under XR the constrained existence problem is not only decidable, but $\NP$-complete. Moreover, when all players are extreme optimists, the problem becomes $\PTIME$-complete. Our algorithmic results apply to all rewards, positive or negative, establishing the first decidable fragment for equilibria in simple stochastic games with terminal objectives without restrictions on strategy types or number of players.
\end{abstract}

\section{Introduction}
Stochastic systems have been used extensively in several areas including  verification~\cite{FKNP11}, learning theory~\cite{AJKS21}, epidemic processes~\cite{Lef81} to name a few. Several real-world systems however do not work with a centralised control. Therefore, modelling using stochastic systems with multiple agents makes for more faithful abstractions of such systems without a centralised control. Some examples of fields in which multi-agents stochastic modelling include cyber physical systems~\cite{SEC16}, distributed and probabilistic computer programs~\cite{dAHJ01}, probabilistic planning~\cite{TKI10}. In such cases, the problem of reasoning about multiple agents with several, often times orthogonal objectives, becomes important. 
For multi-agent systems modelled with stochasticity on the underlying arena, a fundamental question to ask is the existence or finding of an equilibrium.
The most popular equilibria in literature are Nash equilibria~\cite{Nas50}. However, those come with their own downsides. The computational complexity for studying Nash equilibria over multi-agent systems is prohibitively expensive, and even undecidable in the general case, where systems have $10$ or more players~\cite{UW11}. 
Further, even if Nash equilibria could be computed efficiently, they do not faithfully model the agents in real world settings
since they do not consider their tolerance or averseness to risk.

Let us consider a $1$-player game where a protagonist is proposed two options: (a) earning \$1; (b) playing a lottery in which, with probability $\frac{1}{40}$, she gets \$40, and with probability $\frac{39}{40}$, she does not earn anything.
Classically, rational strategies would be maximising the expected payoff. From this perspective, both options yield an expected payoff of \$1, making them equivalent.
This approach is particularly justified when the game represents a scenario that can be repeated many times: the law of large numbers ensures that, in the long run, the average payoff will converge to the expected payoff. However, when the game is played only once, the protagonist may prioritise immediate needs. If she urgently requires \$1, the guaranteed option (a) becomes preferable.

Conversely, if she is a risk-taker or finds herself in a situation where only the \$40 can make a significant difference, she may prefer the high-risk option (b).
Although this choice might appear irrational, it mirrors the behaviour of millions of people who participate everyday in games with a negative expected payoff, driven by the allure of a potentially life-changing win, and generating an annual turnover of USD 536 billions~\cite{GamblingNewspaper23} for the gambling industry.
That industry, on the other hand, operates on a large scale where expected payoff becomes the key metric. 
This contrast underscores the importance of alternative measures to expected payoff that account for each agent's risk tolerance.

\subparagraph*{Risk Measures.}
A \emph{risk measure} captures the perception that a player has of what their payoff will be. In that sense, they generalise the notion of expected payoff.
Various risk measures exist in the literature, and have been used extensively in the field of economics and finance. 
Some of these risk measures include expected shortfall (ES), value at risk (VaR)~\cite{Aue18}, variance~\cite{Bra99}, entropic risk measure (ER)~\cite{FS02}. 

A lot of work has been done in considering these risk measures over MDPs which use variance (along with mean) as a risk-measure~\cite{FK89, PSB22,MT11}, ES~\cite{RRS15,KM18,Meg22} (also referred to as conditional value at risk (CVaR), average value at risk (AVaR), expected tail loss (ETL), and superquantile in literature) and ER~\cite{HM72,BR14,BCMP24}. 
Studying the entropic risk measure in MDPs appears more practical compared to expected shortfall  or using variance-penalised risk-measures. This impracticability of ES and variance-penalised measure in particular is due to the intractable exponential memory~\cite{HK15} and time required to compute optimal strategies~\cite{PSB22}, even for the one agent system of Markov decision processes (MDPs). On the other hand, when the risk measure used is ER, players have optimal positional strategies in MDPs~\cite{How72}, which makes it a prime candidate for consideration in multi-agent settings.

\subparagraph*{Entropic Risk Measure.}
The entropic risk measure is computed by assigning to each agent a risk parameter, i.e., a value $\rho \in \Rb$.
The entropic risk measure of a random variable $X$ is then defined as
$\re_\rho[X] = -\frac{1}{\rho} \log_e \left( \Eb \left[ e^{-\rho X}\right] \right)$. 
If the risk parameter $\rho$ is positive, then more weight will be given to the bad payoffs: the corresponding player can then be considered as risk-averse.
Conversely, players with a negative $\rho$ are more risk-loving.
When $\rho$ tends to $0$, the entropic risk measure converges to the classical expectation $\Eb[X]$.

The game depicted by Figure~\ref{fig:example_gamma} extends the lottery example we discussed earlier. 
Black vertices are stochastic, and the circle vertex is controlled by player $\Circle$.
A play can be seen as an infinite sequence of moves of a token along the edges of the graph, starting from $a$: from a stochastic vertex, it takes one of the outgoing edges with the probabilities indicated on those, and from a vertex controlled by the player, she chooses which edge it takes.
The payoff $40$, $0$, or $1$ is obtained when the terminal vertex $t_1$, $t_2$, or $t_3$ is reached, respectively.
If no terminal vertex is reached, then the payoff is $0$.
Taking the red edge corresponds to option (a): then, her risk entropy is always $1$, for every risk parameter $\rho$.
But if she chooses option (b), that is, if she takes the blue edge, her risk entropy is $\re_\rho[\mu_{\circ}] = -\frac{1}{\rho} \log \left( \Eb \left[ e^{-\rho \mu_{\circ}}\right] \right) = -\frac{1}{\rho} \log \left(  e^{-40\rho } + \frac{39}{40} \right)$.
Both cases are illustrated with red and blue curves in \cref{fig:example_plot}.
The curves cross at abscissa $\rho = 0$, where the entropic risk measure corresponds to the expectation. Note that other strategies are possible if \emph{randomisation} is allowed---the player could, for example, toss a coin and participate in the lottery if the outcome is heads. The perceived reward of randomising between outermost red and blue edges are illustrated in the intermediate cases with mixtures of red and blue in \cref{fig:example_plot}.




 
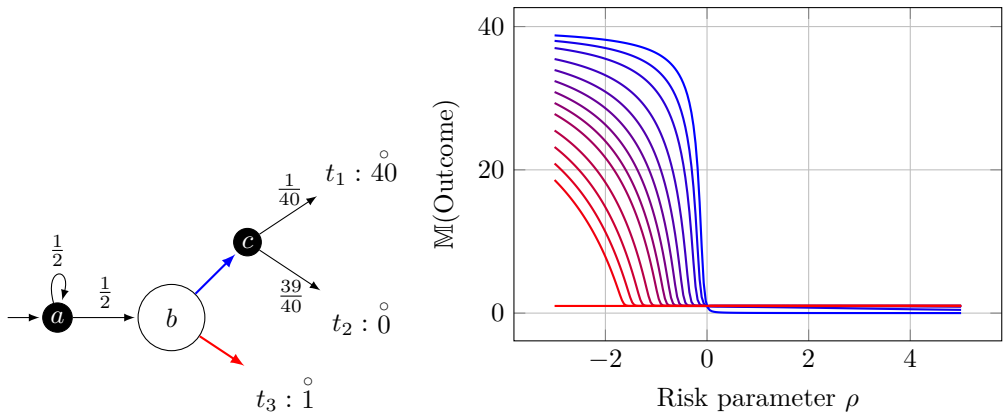
\begin{figure}[h] 
			\centering
            \begin{subfigure}[t]{0.4\textwidth}
			\begin{tikzpicture}[->,>=latex,shorten >=1pt, initial text={}, scale=1, every node/.style={scale=1}]
				\node[initial left, stoch] (a) at (0, 0) {$a$};
				\node[state] (b) at (1.5, 0) {$b$};
                \node[stoch] (c) at (2.5, 1) {$c$};
                \node (t1) at (4, 2) {$t_1:~\stack{\circ}{40}$};
                \node (t2) at (4, 0) {$t_2:~\stack{\circ}{0}$};
                \node (t3) at (3, -1) {$t_3:~\stack{\circ}{1}$};
                \path (a) edge[loop above] node[above] {$\frac{1}{2}$} (a);
				\path (a) edge node[above] {$\frac{1}{2}$} (b);
                \path (b) edge[blue, thick] (c);
                \path (b) edge[red, thick] (t3);
				\path (c) edge node[above] {$\frac{1}{40}$} (t1);
				\path (c) edge node[below] {$\frac{39}{40}$} (t2);
			\end{tikzpicture}
			\caption{A stochastic MDP}
			\label{fig:example_gamma}
            \end{subfigure}
            \begin{subfigure}[t]{0.55\textwidth}
			\begin{tikzpicture}
              \begin{axis}[
                xlabel={Risk parameter $\rho$},
                ylabel={$\re(\text{Outcome})$},
                domain=-3:5,
                samples=200,
                    width=8cm,
                   height=6cm,
                grid=major,
                ]
                \addplot [
                  red!8!blue,
                  thick
                ]
                {-1/x * log2((9/10)*e^(-1*x) + (1/400) * e^(-40*x) + (39/400))/log2(e)};
                \addplot [
                  red!16!blue,
                  thick
                ]
                {-1/x * log2((199/200)*e^(-1*x) + (1/8000) * e^(-40*x) + (39/8000))/log2(e)};
                \addplot [
                  red!24!blue,
                  thick
                ]
                {-1/x * log2((19999/20000)*e^(-1*x) + (1/800000) * e^(-40*x) + (39/800000))/log2(e)};
                \addplot [
                  red!32!blue,
                  thick
                ]
                {-1/x * log2((1999999/2000000)*e^(-1*x) + (1/80000000) * e^(-40*x) + (39/80000000))/log2(e)};
                \addplot [
                  red!40!blue,
                  thick
                ]
                {-1/x * log2((199999999/200000000)*e^(-1*x) + (1/8000000000) * e^(-40*x) + (39/8000000000))/log2(e)};
                \addplot [
                  red!48!blue,
                  thick
                ]
                {-1/x * log2((19999999999/20000000000)*e^(-1*x) + (1/800000000000) * e^(-40*x) + (39/800000000000))/log2(e)};
                \addplot [
                  red!56!blue,
                  thick
                ]
                {-1/x * log2((1999999999999/2000000000000)*e^(-1*x) + (1/80000000000000) * e^(-40*x) + (39/80000000000000))/log2(e)};
                \addplot [
                  red!64!blue,
                  thick
                ]
                {-1/x * log2((199999999999999/200000000000000)*e^(-1*x) + (1/8000000000000000) * e^(-40*x) + (39/8000000000000000))/log2(e)};
                \addplot [
                  red!72!blue,
                  thick
                ]
                {-1/x * log2((199999999999999999/200000000000000000)*e^(-1*x) + (1/8000000000000000000) * e^(-40*x) + (39/8000000000000000000))/log2(e)};
                \addplot [
                  red!80!blue,
                  thick
                ]
                {-1/x * log2((199999999999999999999/200000000000000000000)*e^(-1*x) + (1/8000000000000000000000) * e^(-40*x) + (39/8000000000000000000000))/log2(e)};
                \addplot [
                  red!88!blue,
                  thick
                ]
                {-1/x * log2((199999999999999999999999/200000000000000000000000)*e^(-1*x) + (1/8000000000000000000000000) * e^(-40*x) + (39/8000000000000000000000000))/log2(e)};
                \addplot [
                  red!94!blue,
                  thick
                ]
                {-1/x * log2((199999999999999999999999999/200000000000000000000000000)*e^(-1*x) + (1/8000000000000000000000000000) * e^(-40*x) + (39/8000000000000000000000000000))/log2(e)};
                \addplot [
                  blue,
                  thick
                ]
                {-1/x * log2((1/40) * e^(-40*x) + (39/40))/log2(e)};
                \addplot [
                  red,
                  thick
                ]
                {+1};
              \end{axis}
            \end{tikzpicture}
			\caption{Each curve represents the perceived reward of a player choosing only blue strategy, only red, or  randomising between both strategies. The percieved payoff for a player with risk parameter $\rho \in (-3,5)$ for these strategies are represented.}
			\label{fig:example_plot}
            \end{subfigure}
        \caption{Entropic risk measure}\label{fig:example_re}
\end{figure}
Unfortunately, even for two player zero-sum stochastic games with total-reward objectives (payoff is the sum of the rewards seen along the way), computing optimal strategies can only be done in $\PSPACE$, when the base $e$ is replaced by an algebraic number; and if $e$ is the base of the exponent, then it is decidable only subject to Shanuel's conjecture~\cite{BCMP24}. 
Solving the two-player zero-sum case is a specific case of finding equilibria in two-agent systems where the payoffs of the two agents are exactly the negation of each others and so are the risk parameters of each of the agents.
Therefore, reasoning about multi-agent systems with ER also has potential to be computationally intractable.

\subparagraph*{Extreme Risk Measure.} We introduce a new risk measure called extreme risk measure (XR) to identify tractable risk parameters. 
%
Consider an agent who wishes to maximise the lowest payoff received with positive probability.
In our example, 
by choosing option (a) her only payoff is $\$1$, whereas by choosing option (b), the payoffs that she receives with positive probability are $\$40$ and $\$0$. 
This agent would choose the option (a) since, then, the lowest reward she gets is $\$1$, instead of $\$0$. This would be her choice regardless of the probabilities or if the lottery amount in option (b) is increased.
Such agents can be considered ``extreme pessimists'' because
their perceived payoff can be thought of as the minimum among all the possible payoffs.
Similarly, one can define ``extreme optimists''  whose perceived reward is the best payoff that can be obtained with positive probability.
In the above scenario, an extreme optimist posed with the same options would choose option (b), no matter how small the probability is of receiving that payoff.

Extreme pessimists can be used to model safety-critical agents, where any positive probability of low reward or failure is unacceptable.
On the other hand, extreme optimists model naturally the opponents of such agents.
In a multiplayer setting, they can be an accurate modelling of agents like hackers in a system, who are happy with a small probability of success, or agents that have the possibility to restart their interactions with the same system, so that as long as there is a non-zero probability of achieving a high reward, they are guaranteed to receive that high reward. 


\subparagraph*{Our results.}
We consider the problem of finding equilibria in a multiplayer stochastic game, that is, a game in which the payoffs that the players receive depend on the \emph{terminal vertex} that is reached, and in which an infinite play is associated to the zero payoff vector.

Our contributions are four fold. 
Firstly, we consider the problem of finding equilibria where entropic risk measure is used to determine the perceived reward of each player.  Each player has their own risk-sensitivity parameter, and we wish to find an equilibrium where no player has the incentive to deviate and increase their risk measure. We show that, when the rewards are all non-negative, such an equilibrium always exists.
We conjecture that this remains true when rewards can be negative.
Although some equilibria exist, not all equilibria are made the same, with some equilibria being more desirable than the others. One might want to find an equilibrium that maximises the overall social welfare, or want to minimise it for certain agents. A reasonably general setting is providing an interval for the risk measure for each agent and to check if there is an equilibrium satisfying these constraints. We call this problem \emph{constrained existence problem of risk-sensitive equilibria} (RSEs). 
We show (in \cref{sec:ERM}) that this problem is undecidable when the risk parameters of the players are rational values, with undecidability results extending from the constrained existence problem for Nash equilibria in the work of Ummels and Wojtczak~\cite{UW11}. However, we find restrictions on strategies to recover decidability. 
If we restrict the memory requirements of each player, then for (small) finite memory strategies, we can solve the problem by encoding it using the existential theory of reals with exponentiation, giving us decidability subject to Shanuel's conjecture, and $\PSPACE$ algorithms when the base of the exponents are encoded as small algebraic instances, reminiscent of the two-player zero-sum case by Baier et al.~\cite{BCMP24}. 

Secondly, since the general problem is undecidable, and even in restricted cases, we obtain complexities that are $\PSPACE$ or higher, we pivot to searching for a more tractable risk measure that can be used to find equilibria in multi-agent systems. We define extreme risk measure (XR) as a novel risk measure to consider in multi-agent stochastic systems. We show (in \cref{sec:XR}) that our new definition is robust, since it exactly captures the well-studied entropic risk measure when the risk parameters tend to $\pm \infty$.
We further show the existence of 
such equilibria for games with non-negative rewards. Moreover, there exists a stationary strategy profile that can be algorithmically constructed in polynomial time. We conjecture, again, that this remains true when negative rewards are involved.
One further advantage of XR as a risk measure is that it is indifferent to the exact probabilities of the underlying stochastic model, since it only deals with events that occur with a positive probability and, therefore, can also be used in systems where the underlying probabilities are unknown. 

Thirdly, we show that the constrained existence problem of RSEs is decidable and also $\NP$-complete when the perceived payoff is calculated using XR, where each agent is either an extreme optimist or pessimist. The $\NP$ membership is nontrivial and follows several steps. First, we show that if there is a strategy that satisfies the constraints, then there is a finite abstraction of this strategy. Later, we show that this finite abstraction of the strategy has a polynomial representation. 
With this polynomial representation of the strategy, we show that verifying whether a given polynomially represented strategy is a risk-sensitive equilibrium that satisfies the constraints can also be done in polynomial time. 
Finally, we show that if all players are extreme optimists, this problem is $\PTIME$-complete.
\section{Preliminaries} 
We assume that the reader is familiar with the basics of probability and graph theory. However, we define some concepts for establishing notation. 

\subparagraph*{Probabilities.} Given a (finite or infinite) set of outcomes $\Omega$ and a probability measure $\prob$ over $\Omega$, let $X$ be a random variable over $\Omega$, that is, a mapping $X: \Omega \to \Rb$. We then write $\Eb_\prob[\X]$, or simply $\Eb[X]$, for the expectation of $X$, when it is defined.
Given a finite set $S$, a \emph{probability distribution} over $S$ is a mapping $d: S \to [0,1]$ that satisfies the equality $\sum_{x \in S} d(x) = 1$.
We write $\Supp (d)$ for the \emph{support} of the distribution $d$, that is, the set of elements $x \in S$ such that $d(x) > 0$.

\subparagraph{Risk measures.}
Given a set $\Omega$ of outcomes, a \emph{risk measure} over $\Omega$ is a mapping $M$ which maps a probability measure $\prob$ over $\Omega$ and a random variable $X$ to a real value $M^\prob[X]$.

Sometimes, in the literature, risk measures are expected to have the following three properties: (1) they are \emph{normalised}, i.e., we have $M^\prob[0] = 0$; (2) they are  \emph{monotone}, i.e., the pointwise inequality $X \leq Y$ implies $M^\prob[X] \leq M^\prob[Y]$; and (3) they are \emph{translative}, i.e., $M^\prob[X + c] = M^\prob[X] + c$ for every constant $c$.
In particular, the expectation of a random variable $\Eb$ is a risk measure.
We do not refer to the above properties again and only state them here to remark that all the risk measures we will consider satisfy the above properties.
We remark that the definition of translative sometimes instead refers to satisfying the opposite of the property we define as translative, $M^\prob[X + c] = M^\prob[X] - c$ for every constant $c$. This is a matter of whether we use our risk measure or its negation.


\subparagraph*{Graph, paths, games.}A directed graph $(V,E)$ consists of a set of \emph{vertices} $V$ and a set of ordered pair of vertices, called \emph{edges}, $E$. 
In a directed graph $(V, E)$, for each vertex $u$, we write $E(u)$ to represent the set $E \cap (\{u\} \times V)$.
For simplicity, we often write $uv$ for an edge $(u, v)\in E$.
A \emph{path} in the directed graph $(V, E)$ is a (finite or infinite) word $\pi = \pi_0 \pi_1 \dots$ over the alphabet $V$ such that $\pi_n\pi_{n+1} \in E$ for every $n$ such that $\pi_n$ and $\pi_{n+1}$ exist.
We write $\Occ(\pi)$ for the set of vertices occurring along $\pi$, and $\Inf(\pi)$ for those that occur infinitely often, if there are any.
The prefix $\pi_0 \dots \pi_n$ is written as $\pi_{\leq n}$ or $\pi_{< n+1}$, and the suffix $\pi_n \pi_{n+1} \dots$ is written as $\pi_{\geq n}$ or $\pi_{>n-1}$.
A finite path $\pi = \pi_0 \dots \pi_n$ is \emph{simple} if every vertex occurs at most once along $\pi$.
It is a \emph{cycle} if its last vertex $\pi_n$ is such that $\pi_n\pi_0 \in E$.

\begin{definition}[Game]
    A \emph{game} is a tuple $\Game = (V, E, \Pi, (V_i)_{i \in \Pi}, \p, \mu)$, where we have:
    \begin{itemize}
        \item a directed graph $(V, E)$, called the \emph{underlying} graph of $\Game$;

        \item a finite set $\Pi$ of \emph{players};

        \item a partition $(V_i)_{i \in \Pi \cup \{?\}}$ of the set $V$, where $V_i$ denotes the set of vertices \emph{controlled} by player $i$, and the vertices in $V_?$ are called \emph{stochastic vertices};

        \item a \emph{probability function} $\p: E(V_?) \to [0, 1]$, such that for each stochastic vertex $u$, the restriction of $\p$ to $E(u)$ is a probability distribution;

        \item a mapping $\mu: T \to \Rb^\Pi$ called \emph{payoff function}, where $T$ is the set of \emph{terminal vertices}, i.e. vertices of the graph $(V, E)$ that have no outgoing edges.
        We also write $\mu_i$, for each player $i$, for the function that maps a terminal vertex $t$ to the $i^\text{th}$ coordinate of the tuple $\mu(t)$.
    \end{itemize}
\end{definition}

In a more general framework, payoffs can be assigned to all infinite paths.
Here, we only focus on what is usually called \emph{simple quantitative games}, i.e. games in which the underlying graph contains terminal vertices and the payoffs depend only on which terminal vertex is eventually reached.
We thus extend the mapping $\mu$ to the set $(V \setminus T)^\omega \cup (V \setminus T)^* T$ by defining $\mu(v_1 \dots v_k t) = \mu(t)$, and $\mu(v_1 v_2 \dots) = (0)_{i \in \Pi}$ (if no terminal vertex is reached, everyone gets the payoff $0$).
A game is \emph{Boolean} if all payoffs belong to the set $\{0, 1\}$.

An \emph{initialised game} is a tuple $(\Game, v_0)$, usually written $\Game_{\|v_0}$, where $v_0 \in V$ is an \emph{initial vertex}.
In what follows, when the context is clear, we use the word \emph{game} also for an initialised game.
We often assume that we are given a game $\Game_{\|v_0}$ and implicitly use the same notations as in the definition above.

\subparagraph*{Histories and plays.} We call \emph{play} a path in the underlying graph that is infinite, or whose last vertex is terminal.
Other paths are called \emph{histories}.
We will then use the notations $\Hist(\Game)$ to denote finite paths in the graph of the game, and $\Plays(\Game)$ to denote both finite and infinite paths. For a history $h = h_0 \dots h_n$, we write $\last(h) = h_n$.
We will also write $\Hist_i(\Game)$ for the set of histories whose last vertex is controlled by player $i$.
A history or play in an initialised game $\Game_{\|v_0}$ is a history or play in $\Game$ whose first vertex is $v_0$.

\begin{definition}[Markov decision process, Markov chain]
    A (initialised or not) \emph{Markov decision process} is a game with one player.
    A \emph{Markov chain} is a game with zero player.
\end{definition}

\paragraph*{Strategies, and strategy profiles}


    In a game $\Game_{\|v_0}$, a \emph{strategy} for player $i$ is a mapping $\sigma_i$ that maps each history $hu \in \Hist_i(\Game_{\|v_0})$ to a probability distribution over $E(u)$.
The set of possible strategies for player $i$ in $\Game_{\|v_0}$ is written as $\Strat_i(\Game_{\|v_0})$.
A path $\pi_0 \pi_1 \dots$ (be it a history or a play) is \emph{compatible} with the strategy $\sigma_i$ if for each $k$ such that $\pi_k \in V_i$, we have the probability that the strategy $\sigma_i$ proposes $h_{k+1}$ after  history $h_{\leq k}$ is positive, that is, $\sigma_i(h_{\leq k})(h_{k+1}) > 0$.
A \emph{strategy profile} for a subset $P \subseteq \Pi$ is a tuple $(\sigma_i)_{i \in P}$.
A strategy profile for the set $P$ of players is written $\bsigma_P$, or simply $\bsigma$ when $P = \Pi$. We also write $\bsigma_{-i}$ for $\bsigma_P$ where $P = \Pi \setminus \{i\}$.
Similarly, we use $(\sigma_{-i}, \sigma'_i)$ to denote the strategy profile $\btau$ defined by $\tau_i = \sigma'_i$ and $\tau_j = \sigma_j$ for $j \neq i$. 
We sometimes write $\bsigma(hv)$ to mean $\sigma_i(hv)$ where $i$ is the player controlling $v$, or $\p(v)$ when $v \in V_?$.
We sometimes also write $V_{-i}$ instead of $\bigcup_{j \neq i} V_j$. 

For some history $h$, and a strategy $\sigma_i$, we define the strategy truncated to a history $h$, written $\sigma_{i\| hv}$, as the strategy $\sigma'_i: h' = \sigma_i(hh')$ in the game $\Game_{\|v}$.

A strategy profile $\bsigma_{-i}$ in the game $\Game_{\|v_0}$ defines an initialised Markov decision process $\Game_{\|v_0}(\bsigma_{-i})$, where the vertices of the (infinite) underlying graph are the histories of $\Game_{\|v_0}$ and the edges are added from $hu$ to each the history $huv$ iff  $uv \in E$.
Similarly, a strategy profile $\bsigma$ for $\Pi$ defines an initialised Markov chain $\Game_{\|v_0}(\bsigma)$.
Thus, it also defines a probability measure $\prob_\bsigma$ over plays --- which turns the payoff functions $\mu_i$ into random variables.

\subparagraph{Pure, stationary, and positional strategies.}
We say that a strategy $\sigma_i$ is \emph{pure} when for each history $hu$, there is a vertex $v$ such that $\sigma_i(hu)(v) = 1$; then we often just write $\sigma_i(hu) = v$. 
We say that $\sigma_i$ is \emph{stationary} when for every two histories $hu, h'u \in \Hist_i(\Game_{\|v_0})$, we have $\sigma_i(hu) = \sigma_i(h'u)$.
In that case, we sometimes assume that strategy $\sigma_i$ is defined in every game $\Game_{\|u}$ and simply write $\sigma_i(u)$ for $\sigma_i(hu)$.
Finally, the strategy $\sigma_i$ is \emph{positional} when it is pure and stationary.
Those concepts are naturally generalised to strategy profiles.

\subparagraph*{Memory structures.}
A \emph{memory structure} for player $i$ is a tuple $(S_i, s_0, \delta_i, \nu_i)$, where $S_i$ is a finite set of \emph{memory states}, where $s_0 \in S_i$ is an \emph{initial state}, where $\delta_i$ is a \emph{memory-update mapping} that maps each pair $(s, v) \in S_i \times V$ to a memory state $s'$, and where $\nu_i$ is an \emph{output mapping} that maps each pair $(s, v) \in S_i \times V_i$ to a distribution $d$ over $E(v)$.
The memory-update mapping can be extended to a mapping $\delta_i^*: \Hist(\Game_{\|v_0}) \to S_i$ with $\delta_i^*(\epsilon) = m_0$ and $\delta_i^*(hu) = \delta_i(\delta_i^*(h), u)$ for each history $hu$.
The memory structure then induces a strategy $\sigma_i$ defined by $\sigma_i(hu) = \nu_i(\delta^*_i(h), u)$ for each history $hu \in \Hist_i\Game_{\|v_0}$.
A strategy induced by a memory structure is called \emph{finite-memory strategy}.
Note that stationary strategies are exactly the strategies that are induced by a memory structure with $|S_i| = 1$.

A tuple of finite-memory strategy profiles for each player is a \emph{finite-memory strategy profile}  is given by a similar object $(S, s_0, \delta, \nu)$ to memory structure, where by replacing $\nu$ by its restriction to $S \times V_i$, we obtain a memory structure that induces the strategy $\sigma_i$.

\subparagraph{Risk-sensitive equilibria, constrained existence problem.}
In multiplayer stochastic games, we wish to study generalisations of the classical Nash equilibria where the expectation is replaced by other risk measures.
Such generalisations are called \emph{risk-sensitive equilibria}~\cite{Now05}. We define this for games played over graphs

When $M$ is a risk measure and $\bsigma$ is a strategy profile, we write $M(\bsigma)$ for $M^{\prob_\bsigma}$.

\begin{definition}[Risk-sensitive equilibrium]
    Let $\Game_{\|v_0}$ be a game, and let $\bM = (M_i)_{i \in \Pi}$ be a tuple of risk measures.
    Let $\bsigma$ be a strategy profile in $\Game_{\|v_0}$, let $i$ be a player, and let $\sigma'_i$ be a strategy for player $i$, called \emph{deviation} of player $i$ from $\bsigma$.
    The deviation $\sigma'_i$ is \emph{profitable} with regards to the risk measure $M_i$ if we have $M_i(\bsigma_{-i}, \sigma'_i)[\mu_i] > M_i(\bsigma)[\mu_i]$
    The strategy profile $\bsigma$ is a $\bM$-\emph{risk-sensitive equilibrium}, or $\bM$-RSE, if no player $i$ has a profitable deviation from $\bsigma$ with regards to $M_i$.
\end{definition}


The following problem is the main focus throughout our paper.

\begin{question}[Constrained existence of risk-sensitive equilibria]
    Given a game $\Game_{\|v_0}$, a tuple of risk measures $\bM$, and two payoff vectors $\bx, \by \in \Qb^\Pi$, does there exist a $\bM$-RSE $\bsigma$ in $\Game_{\|v_0}$ such that for each $i \in \Pi$, we have $x_i \leq M_i(\bsigma)[\mu_i] \leq y_i$?
\end{question}

To turn this problem into an algorithmic decision problem, we need to restrict it to some specific sets of risk measures that can be finitely encoded. 
That is what we do in the sequel of this paper, with the \emph{entropic risk measure}, and later with the \emph{extreme risk measure}.


\section{Entropic risk measure}\label{sec:ERM}
The entropic risk measure is a measure of the perceived payoff, which depends on the aversion or inclination of the player toward risk through the exponential utility function. 
It is defined using a \emph{risk parameter}, i.e. a real value $\rho\in\Rb \setminus \{0\}$: large positive values indicate risk-averseness, large negative values risk-inclination. To see a visual representation of the entropic risk measure, see \cref{fig:example_re} in the introduction.

\begin{definition}[Entropic risk measure]
Given a risk parameter $\rho$, the \emph{entropic risk measure} is defined for every probability measure $\prob$ and random variable $X$ as
$$\re_{\rho}^\prob[X] = -\frac{1}{\rho} \log_e \left( \Eb^\prob \left[ e^{-\rho X}\right] \right).$$
For computational reasons, this definition is generalised by allowing every base $\beta > 1$ instead of Euler's constant. The \emph{entropic risk measure with base $\beta$} is then defined by: 
$$\re^\prob_{\beta\rho}[X] = -\frac{1}{\rho} \log_\beta \left( \Eb^\prob \left[ \beta^{-\rho X}\right] \right).$$
\end{definition}

The three parameters $\prob$, $\rho$ and $\beta$ can be omitted when they are clear from the context.

\begin{remark}
\begin{itemize}
    \item For every $\beta$ and $\rho$, the entropic risk measure $\RM_{\beta\rho}$ is a risk measure.

    \item By enabling any base $\beta$, we obtain a definition that is more general only on a computational level, since handling Euler's constant may not be equivalent to handling rational values.
    Baring computational concerns, these definitions with different bases are equivalent, since for every $\beta$ we have $\RM_{\beta\rho} = \RM_{e\rho'}$, where $\rho' = \rho \log_e(\beta)$.

    \item The above definition implies that for $\rho = 0$, the function is not defined.
    However, it is known that for all $\prob$, $\beta$ and $X$, the quantity $\RM_{\rho}$ converges to $\Eb[X]$ when $\rho$ tends to $0$ (see e.g.~\cite{PDM20}).
    Therefore, we henceforth assume that $\RM_{0}[X] = \mathbb{E}[X]$ to make risk entropy defined for all finite risk parameters $\rho$.
\end{itemize}
\end{remark}

When we are given a profile $\brho = (\rho_i)_{i \in \Pi}$ of risk parameters, we will sometimes write $\M_{\beta\brho}[\mu]$ for the tuple $\left(\M_{\beta\rho_i}[\mu_i]\right)_{i \in \Pi}$.
Risk entropy defines a family of RSEs, namely the $(\M_{\beta\rho_i})_i$-RSEs, that we also call \emph{$(\beta, \brho)$-entropic risk-sensitive equilibria}, or $(\beta, \brho)$-ERSEs.
The following theorem states the existence of such an RSE that uses no randomness in its strategy profile, in cases where all the payoffs are non-negative. 

\begin{theorem}[Existence of ERSE]\label{thm:existanceRSE}
    Let $\Game_{\|v_0}$ be a simple stochastic game with only non-negative payoffs.
    Then, there exists a (pure) $(\beta,\rho)$-ERSE over $\Game_{\|v_0}$.
\end{theorem}

\begin{proof}
 Pure Nash equilibria always exists in a stochastic multi-player games with prefix-closed Boolean objectives~\cite[Theorem 3.10]{Umm10} (a correction of an existing proof~\cite{CMJ04}). It is known that simple stochastic games where rewards are all positive (or all negative) can be converted into a game with reachability objectives such that if there is an NE in one, there is an NE in the converted game with the reachability objective. Indeed, if all the rewards are positive, we can always scale the rewards for each player of a stochastic game to ensure they are in the unit interval $[0,1]$. If the rewards are within the unit interval, then for terminals with reward $p$, we can instead add a probabilistic node that reaches this terminal vertex with probability $p$. 
Therefore, with the same result, Nash equilibria always exist in simple stochastic games with non-negative rewards on the terminals. 

Then, we can conclude our theorem using the following lemma.

\begin{restatable}[App.~\ref{lemma:RSEtoQSSG}]{lemma}{RSEtoQSSG}\label{lemma:RSEtoQSSG}
Given a game $\Game_{\|v_0}$ and a tuple $\brho \in \Rb^\Pi$, there exists a game $\Game'_{\|v_0}$ with the same underlying graph, player set, and probability function (but possibly different payoff function), such that the $(\beta,\rho)$-ERSEs in $\Game_{\|v_0}$ are exactly the Nash equilibria in $\Game'_{\|v_0}$. \qedhere
\end{restatable}
\end{proof}

We conjecture that this result remains true when we remove the guarantee that rewards are non-negative.
We now turn to the constrained existence problem of $\tpl{\beta,\brho}$-ERSEs.
Unfortunately, it is undecidable in the general case.

\begin{proposition}\label{proposition:Undecidable}
    The constrained existence problem of $\tpl{\beta,\brho}$-ERSEs with $\brho \in \Qb^{\Pi}$ is undecidable, even for any fixed value of $\beta$, for $\brho = (0)_i$, and with only nonnegative payoffs. 
\end{proposition}

\begin{proof}
         The undecidability of the constrained existence problem follows from the work of Ummels and Wojtczak~\cite[Theorem 4.9]{UW11} where they show the undecidability of the constrained existence problem for Nash equilibria in the setting with 10 or more players. Since Nash equilibria is a specific instance of the setting of ERSEs where the risk parameters $\brho$ is $0$ for each player, the undecidability of our setting follows. 
\end{proof}

We therefore turn our attention to the constrained existence problem when the class of strategies considered is restricted. 

\begin{restatable}[App.~\ref{app:ERRSErestricted}]{theorem}{stationaryRSE}\label{thm:ERRSErestricted}
The constrained existence problem of $(\beta,\brho)$-ERSEs, in quantitative simple stochastic games:
\begin{enumerate}
    \item remains undecidable when players are restricted to pure strategies;\label{itm:ERRSEitmundec} 
    \item is decidable when players are restricted to stationary strategies\label{itm:ERRSEdecidable}
\begin{enumerate}
        \item subject to Shanuel's conjecture if $\beta = e$ and the risk-parameters $\rho_i$ are algebraic;\label{itm:ERRSEitmShanuel}
        \item in $\PSPACE$ if the risk parameters and the base $\beta$ are algebraic, in which case it is also $\NP$-hard and $\SQRTSUM$-hard.\label{itm:ERRSE:PSPACE}
        The $\NP$ lower bound also holds for the case where strategies are restricted to positional strategies.
    \end{enumerate}
\end{enumerate}
\end{restatable}

\begin{proof}[Proof Sketch]
    The undecidability of the case where pure strategies are considered is inherited from Nash equilibria~\cite[Theorem~4.9]{UW11}, since the reduction uses only pure strategies. 
        The decidability of this stationary case is reminiscent of similar results for the two-player zero-sum case, which was recently studied in the work of Baier et al.~\cite{BCMP24}.
    However, the techniques used are quite different and also require inspiration from the work of Ummels and Wojtczak~\cite[Theorem 4.5, Theorem 4.6]{UW11}, with significant modifications. 
    We write formulas in the existential theory of reals ($\exists\Rb$) which puts them in $\PSPACE$. 
    For the case $\beta = e$, this formula can be written in the existential theory of reals with exponentiation, which is decidable subject to Shanuel's conjecture, which is a well-known conjecture in the field of transcendental number theory~\cite{Lan66}. The lower bounds of $\NP$-hardness and $\SQRTSUM$-hardness  also follow from the works of Ummels and Wojtczak~\cite[Theorem~4.4,Theorem~4.6]{UW11}. The exact complexity of $\SQRTSUM$ (deciding, given a set $\{a_1, \dots, a_n\} \subseteq \Nb$ and an integer $t$, whether we have $\sum_i \sqrt{a_i} \leq t$) is open and is known to lie in the polynomial hierarchy and in the fourth level of the counting hierarchy~\cite{AKBM06}. 
\end{proof}

\section{Extreme risk measure}\label{sec:XR}
This section introduces a new risk measure that provides a tractable alternative to existing risk measures available in the literature. We provide a simpler, yet robust, alternative that allows us later to tackle the constrained existence of equilibria in multiplayer settings.


Let us consider a random variable $X$ that ranges over $\Rb$. 
The \emph{pessimistic risk measure} of $X$ is the highest value $x$ such that $X$ almost-surely takes a value above $x$.
When $X$ takes finitely many values, that corresponds to the least value that it takes with positive probability.
In probability theory, that measure is sometimes referred to as \emph{essential infimum}, written $\essinf$.
The definition of \emph{optimistic risk measure} is symmetric.

\begin{definition}[Optimistic, pessimistic risk measure]
The pessimistic risk measure of a random variable $X$ is defined by
$\pexp[X] = \essinf(X) = \sup \{x \in X ~|~ \prob(X \geq x) = 1\}$.
Analogously,  the \emph{optimistic risk measure} of $X$ is $\oexp[X] = \esssup(X) =  \inf \{x \in X ~|~ \prob(X \leq x) = 1\}$.    
\end{definition}


When we are given a game $\Game_{\|v_0}$, we can assign an risk measure for each player by defining a partition $(P, O)$ of $\Pi$, where the set $P$ represents the set of players that are \emph{pessimists}, whose perceived payoffs are defined by the pessimistic risk measure, while $O$ represents the \emph{optimists}, who intend to maximise their optimistic risk measure.
For convenience, we group both measures under the umbrella term \emph{extreme risk measure (XR)}, and often assume that $(P, O)$ is given; 
then, we write $\X_i$ for $\pexp$ when $i \in P$, and for $\oexp$ when $i \in O$.
Since each player $i$ is usually interested only in the risk measure of their own payoff, we will also write $\X_i(\bsigma)$ for the quantity $\X_i(\bsigma)[\mu_i]$.
We define \emph{extreme risk-sensitive equilibria}, or XRSEs for short, as $(\X_i)_i$-RSEs.


\subsubsection*{Equivalence to limit of Entropic Risk Measure}
We show that our definition of extreme risk measure corresponds to the limit cases of entropic risk measure.
Observe that in \cref{fig:example_re}, following the  blue strategy, the only payoffs that are obtained with positive probability were $40$ and $0$, which are also the limits of the risk entropy when $\rho$ tends to infinite values.
On the other hand, in the red strategy, the only payoff obtained with positive probability is payoff $1$. Although the payoff $0$ is possible since the play $a^\omega$ is compatible with every strategy, this play must be ignored since it is realised with probability $0$.

\begin{restatable}[App.~\ref{app:RE=PEorOE}]{theorem}{REisPEOE}\label{thm:RE=PEorOE}
    Let $X$ be a random variable that ranges over $\Rb$, and let $\beta > 1$.

    \begin{itemize}
        \item The limit risk entropy of $X$ when $\rho$ tends to $+\infty$ exists and is equal to the pessimistic risk measure, that is, we have $\lim_{\rho \to +\infty} \re_{\beta\rho} [X] = \pexp[X]$.    
       \item Similarly, the limit risk entropy of $X$ when $\rho$ tends to $-\infty$ exists and is equal to the pessimistic risk measure, that is, we have $\lim_{\rho \to -\infty} \re_{\beta\rho} [X] = \oexp[X]$.
    \end{itemize}
\end{restatable}

\subsubsection*{Extreme risk-sensitive equilibria exist}
We now answer a fundamental question about equilibria, which is if one always exists. 
We show that (stationary) XRSEs are guaranteed to exist in games with only non-negative rewards, similarly as ERSEs.
But our proof does not rely on the same arguments, and we instead give a constructive proof.

\begin{restatable}[App.~\ref{app:XRSEexists}]{theorem}{XRSEexists}\label{thm:XRSEexists}
    Let $\Game_{\|v_0}$ be a game with only non-negative rewards, and let $(P,O)$ be a partition of $\Pi$.
    Then, there exists a stationary XRSE in $\Game_{\|v_0}$.
    Moreover, there exists an algorithm that, given such a game, outputs the representation of such an XRSE in time $\Oh(m^2 p)$, where $m$ is the number of edges, and $p$ the number of \emph{pessimistic} players.
\end{restatable}

\begin{proof}[Proof sketch.]
    Our algorithm generates an XRSE by constructing a decreasing sequence $E = E_0, E_1, \dots$ of sets of edges, and considering, for each $k$, the stationary strategy profile that randomises between all the outgoing edges in $E_k$ from all vertices.
    
    Let us illustrate it with the game depicted by Figure~\ref{fig:ex_extreme1}, which involves two pessimists, player $\Circle$ and player $\Square$.
    In that game, both players want to leave the cycle, but each of them would prefer the player to leave.
    If we first consider the strategy profile that always randomises between all the available edges, then both terminal vertices are reached with positive probability, and it is almost sure that one of them is reached: both players get therefore risk measure $1$.
    Then, player $\Square$ (and symmetrically player $\Circle$) has a profitable deviation by refusing to leave the cycle, and always going back to the vertex $a$.
    Note that player $\Circle$ cannot detect such a deviation of strategy, since she does not have access to the internal coins tossed by player $\Square$. 
    Then, we remove the edge $bt_2$ (or $at_1$). This results in a set of edges where player $\Square$ gets the payoff $2$, and player $\Circle$ cannot get more than $1$, ensuring that the new strategy profile that we obtain is a (stationary) XRSE.
\end{proof}

Like in the case of ERSEs, we conjecture that existence, and even existence of a stationary strategy profile, remain true in the general case.

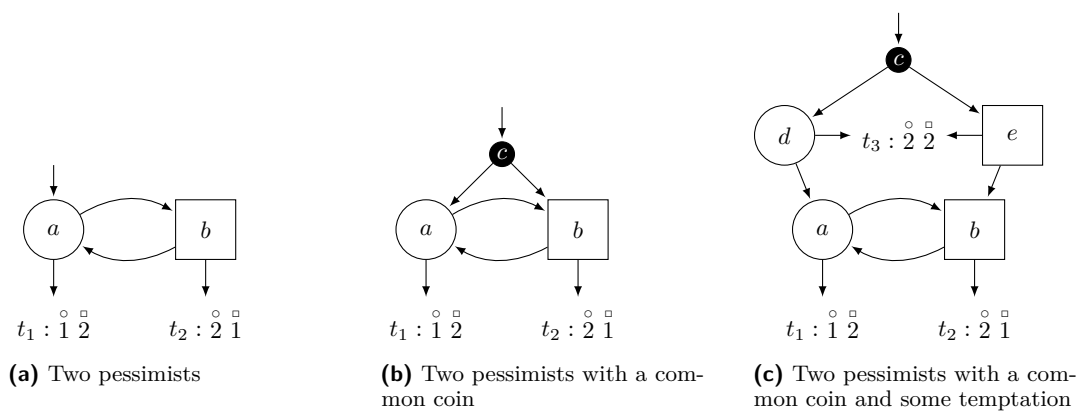
\begin{figure}[h] 
			\centering
            \begin{subfigure}[t]{0.3\textwidth}
			\begin{tikzpicture}[->,>=latex,shorten >=1pt, initial text={}, scale=1, every node/.style={scale=0.9}]
				\node[initial above, state] (a) at (0, 2) {$a$};
				\node[state, rectangle] (b) at (2, 2) {$b$};
                \node (t1) at (0, 0.75) {$t_1:~\stack{\circ}{1}~\stack{\square}{2}$};
                \node (t2) at (2, 0.75) {$t_2:~\stack{\circ}{2}~\stack{\square}{1}$};
                \path (a) edge[bend left] (b);
				\path (b) edge[bend left] (a);
                \path (a) edge (t1);
                \path (b) edge (t2);
			\end{tikzpicture}
			\caption{Two pessimists}
			\label{fig:ex_extreme1}
            \end{subfigure}
            \hfill
            \begin{subfigure}[t]{0.3\textwidth}
			\begin{tikzpicture}[->,>=latex,shorten >=1pt, initial text={}, scale=1, every node/.style={scale=0.9}]
                \node[initial above, state, stoch] (c) at (1, 3) {$c$};
				\node[state] (a) at (0, 2) {$a$};
				\node[state, rectangle] (b) at (2, 2) {$b$};
                \node (t1) at (0, 0.75) {$t_1:~\stack{\circ}{1}~\stack{\square}{2}$};
                \node (t2) at (2, 0.75) {$t_2:~\stack{\circ}{2}~\stack{\square}{1}$};
                \path (c) edge (a);
                \path (c) edge (b);
                \path (a) edge[bend left] (b);
				\path (b) edge[bend left] (a);
                \path (a) edge (t1);
                \path (b) edge (t2);
			\end{tikzpicture}
			\caption{Two pessimists with a common coin}
			\label{fig:ex_extreme2}
            \end{subfigure}
            \hfill
            \begin{subfigure}[t]{0.3\textwidth}
			\begin{tikzpicture}[->,>=latex,shorten >=1pt, initial text={}, scale=1, every node/.style={scale=0.9}]
                \node[initial above, state, stoch] (c) at (1, 4.25) {$c$};
                \node[state] (d) at (-0.5, 3.25) {$d$};
                \node[state, rectangle] (e) at (2.5, 3.25) {$e$};
				\node[state] (a) at (0, 2) {$a$};
				\node[state, rectangle] (b) at (2, 2) {$b$};
                \node (t1) at (0, 0.75) {$t_1:~\stack{\circ}{1}~\stack{\square}{2}$};
                \node (t2) at (2, 0.75) {$t_2:~\stack{\circ}{2}~\stack{\square}{1}$};
                \node (t3) at (1, 3.25) {$t_3:~\stack{\circ}{2}~\stack{\square}{2}$};
                \path (c) edge (d);
                \path (c) edge (e);
                \path (d) edge (t3);
                \path (e) edge (t3);
                \path (d) edge (a);
                \path (e) edge (b);
                \path (a) edge[bend left] (b);
				\path (b) edge[bend left] (a);
                \path (a) edge (t1);
                \path (b) edge (t2);
			\end{tikzpicture}
			\caption{Two pessimists with a common coin and some temptation}
			\label{fig:ex_extreme3}
            \end{subfigure}
        \caption{Some games involving two pessimistic players}\label{fig:ex_extreme}
\end{figure}

\section{Constrained existence of extreme risk-sensitive equilibria}\label{sec:NPComplete}
We now study the computational complexity of the constrained existence problem of XRSEs.
The main result of this section is the following theorem, which proves that, contrary to the same problem with ERSEs, it is a decidable fragment of the constrained existence of RSEs.

\begin{theorem}\label{thm:NPcomplete}
    The constrained existence problem for XRSEs is $\NP$-complete and is $\NP$-hard even when all players are pessimistic and all rewards are non-negative.
\end{theorem}

First, we prove \cref{thm:memorysmall}, which shows that if there is an XRSE then there is one that uses finite memory. We show that this finite-memory strategy profile can be described only using polynomial size, which in turn proves $\NP$ membership (\cref{lemma:np_easy}).

Later, we consider the problem of XRSEs when the players are restricted to pure, stationary or positional strategies. We show that in all the above cases, the problem remains $\NP$-complete. The upperbound is similar to the general case, but the lower bound is shown in \cref{lemma:np_hardness} by showing a reduction from $\THREESAT$ to the constrained existence problem. 
\begin{restatable}{theorem}{restrictedStrategies}\label{thm:infinite_rho_restricted_strategy_np_easy}
    The constrained existence problem of XRSEs is also $\NP$-complete when the players are restricted to positional, stationary,  or pure strategies. 
\end{restatable}

Finally, we show in the following theorem that when all players are optimistic, the problem becomes $\PTIME$-complete. 
\begin{restatable}{theorem}{PTIMEcompleteThm}\label{thm:PTIMEcomplete}
    The constrained existence problem of XRSE is $\PTIME$-complete when all players are optimists, that is, when $P=\emptyset$.
\end{restatable} 
We dedicate the rest of the section to proving these three results. 
\subsection{Membership in $\NP$}
$\NP$-membership is a consequence of the fact that when an XRSE exists, there also exists one with the same extreme risk measures that uses finite memory, with a number of states that is polynomial in the size of the game.
Let us therefore illustrate, with examples, how and why memory is required in such XRSEs.
We consider the following constrained existence question and analyse the same question on three example graphs. 
\begin{quote}$(*)$
   Is there an XRSE in the game in which both players have a risk measure $1$?
\end{quote}

\subparagraph*{Game in \cref{fig:ex_extreme1}.}Let us consider the game in \cref{fig:ex_extreme1} again.
The answer to Question~$(*)$ here is \emph{no}.
Intuitively, such an XRSE would require at least two plays of positive probability: one that ends in $t_1$, and one that ends in $t_2$.
For those two plays to occur with positive probability, the strategy profile must proceed to a randomised action at vertex $a$ or $b$: i.e., one of the players, at some point of time, must toss a coin to give the payoff $1$ to player $\Circle$ in one case and to $\Square$ in the other.
But then, since that player is the only one that can see that coin, they have a profitable deviation by lying about the outcome, and always choose the option that gives them the best payoff.
More randomisation will not help: as long as one of the players randomises, be it once, several times, or infinitely often, they have an incentive to deviate and stay in the cycle and wait until the other player leaves.

\subparagraph*{Game in \cref{fig:ex_extreme2}.} Consider now a slight modification, as shown in~\cref{fig:ex_extreme2}.
There, the first player that plays is determined at random by the edge that is taken from an initial stochastic vertex.
The answer to Question $(*)$ for this game is \emph{yes}. 
The random choice on which player gets payoff $1$ is decided by the stochastic vertex. Since both players can see which edge is taken from there, this serves as a source of unbiased randomness based on which they act. 
For example, it can be decided that if the play visits the vertex $a$ immediately after $c$, then player $\Circle$ must visit the terminal $t_1$, and similarly, if it visits the vertex $b$, then player $\Square$ must visit the terminal $t_2$. If the edge $ab$ is taken, player $\Square$ punishes player $\Circle$ by always going back to $a$, and vice versa.
In other words, the stochastic vertex provides the players with a common coin.

\subparagraph*{Game in \cref{fig:ex_extreme3}.} Finally, consider the game depicted in~\cref{fig:ex_extreme3}.
Here, both players $\Circle$ and $\Square$ have the possibility of deviating to a terminal with payoff $2$ in one play.
The stationary strategy profile in which from vertex $d$, player $\Circle$ goes from $d$ to $a$ and then to $t_1$, and in which player $\Square$ goes from $e$ to $b$ and then to $t_2$, is therefore not an XRSE: both players have a profitable deviation that goes to terminal $t_3$.
But the answer to Question~$(*)$ still remains \emph{yes}!
If, from vertex $d$, player $\Circle$ goes from vertex $d$ to $a$ and then to $b$, from which player $\Square$ leaves to $t_2$, and symmetrically, from vertex $e$, player $\Square$ goes to $a$ through $b$ from which player $\Circle$ goes to $t_1$, then that strategy profile is an XRSE, in which everyone gets the risk measure $1$. 
This is because a player has a profitable deviation only if they can play in a way that guarantees them a risk measure better than $1$, i.e., that guarantees them \emph{almost surely} a payoff greater than~$1$ by deviating.
If there remains a play that occurs with nonzero probability and offers a lower reward, then the player does not increase their risk measure.  
Therefore, an XRSE where player $\Circle$ gets the extreme risk measure $1$ only needs to have one play with positive probability in which she gets the payoff $1$, \emph{and} in which she cannot increase her payoff by deviating. We say that such a play \emph{anchors} that player. In our example, the play $cebat_1$ anchors player $\Circle$. 

We see in this last example that memory is required to remember either the subset of players that are being anchored, or if a player has deviated from the strategy and must be punished. 
Given one or more players that are being anchored, the memory state of any of the players does not change unless either a player deviates or, more importantly, randomisation occurs. When randomisation occurs, the set of players that are anchored in each of the plays is a subset of the set of players anchored before this play \emph{split}.
In our examples, the set of players that are anchored at $c$ is both $\Circle$ and $\Square$, and it immediately splits.
After the splits, when we have only one player to anchor, the players can follow a positional strategy profile; and similarly when one player deviates and must be punished the players can follow a positional strategy profile.


We prove a theorem that bounds the amount of memory required by a strategy to a polynomial in the number of players and vertices in the game. 
\begin{restatable}[App.~\ref{app:memorysmall}]{theorem}{memorysmall}\label{thm:memorysmall}
    Let $\bsigma$ be an XRSE in the game $\Game_{\|v_0}$ with $n$ vertices and $p$ players,  and a partition $(P, O)$ of player $\Pi$.
    Then, there exists a finite-memory XRSE $\bsigma^\star$ with at most $3np-2n+p+1$ many memory states, and such that $\X(\bsigma^\star) = \X(\bsigma)$. Furthermore, if $\bsigma$ is pure, then there is such a strategy profile $\bsigma^\star$ that is pure.
\end{restatable}

\begin{proof}[Proof sketch]
We prove this theorem by formalising the idea of \emph{anchoring plays}.
To do so, we define a labelling function $\Lambda$, that maps each history $h$ compatible with $\bsigma$ to the set of players that is, after the history $h$, currently being \emph{anchored}. In \cref{lm:Lambda}, we show the existence of such a labelling, with some properties. 
In the sequel, we write $\bz$ to represented the risk measure of each player in the strategy profile $\bsigma$, that is, $\bz = (z_i)_i = \X(\bsigma)$.

     \begin{restatable}[The labelling $\Lambda$, App.~\ref{app:memorysmall}]{lemma}{finiteMemAbstraction}\label{lm:Lambda}
        There exists a labelling $\Lambda$ that maps each history $h \in \Hist\Game_{\|v_0}$ compatible with $\bsigma$ to a set $\Lambda(h) \subseteq \Pi$, such that for each such $h$, if we write $\{v_1, \dots, v_k\} = \Supp(\bsigma(h))$, the labelling $\Lambda$ satisfies the following properties.
        \begin{enumerate}
            \item\label{itm:splitsetsanchorwithouti} If the vertex $\last(h)$ is stochastic, or belongs to some player $i \not\in \Lambda(h)$, then the sets $\Lambda(hv_1), \dots, \Lambda(hv_k)$ form a partition of $\Lambda(h)$.

            \item\label{itm:splitsetsanchorwithi} If the vertex $\last(h)$ belongs to some player $i \in \Lambda(h)$, then the sets $\Lambda(hv_1) \setminus \{i\}, \dots, \Lambda(hv_k) \setminus \{i\}$ along with $\{i\}$ form a partition of $\Lambda(h) \setminus \{i\}$, and $i$ belongs to all sets $\Lambda(hv_1), \dots, \Lambda(hv_k)$.

            \item\label{itm:optimistanchor} For each optimistic player $i \in \Lambda(h)$, we have $\X_i(\bsigma_{\|h}) = z_i$.
            
            \item\label{itm:pessimistanchor} For each pessimistic $i \in \Lambda(h)$, for all strategies $\tau_i$ of player $i$, we have $\X_i(\bsigma_{-i\|h}, \tau_i) \leq z_i$.

            \item\label{itm:nosplit} If there is a successor $v_\l$ such that $\Lambda(hv_\l) = \Lambda(h)$, then all other successors $v_{\l'}$ are such that $\X_i(\bsigma_{\|hv_{\l'}}) < z_i$ for each optimist $i \in \Lambda(h)$, and there exists $\tau_i$ with $\X_i(\bsigma_{-i\|hv_{\l'}}, \tau_i) > z_i$ for each pessimist $i \in \Lambda(h)$.
        \end{enumerate}
    \end{restatable}

With such a labelling $\Lambda$, we later show that there are at most $3p-2$ subsets $A$ such that $\lambda(h) = A$ for some history $h$ by an inductive argument (\cref{prop:combinatorial} in \cref{app:memorysmall}). 
We use subsets in the range of the function $\Lambda$  to create $3p-2$ memory states for each of the $n$ vertices to remember the anchoring plays at that vertex of the play, including one extra memory-state for the empty subset. In addition, the memory states also include $p$ punishing  strategies, one for each player, adding up to the number $3np-2n+p+1$. 
We construct a strategy $\bsigma^\star$ from $\Lambda$ that uses only these memory states defined above.
\end{proof}

Finally, using \cref{thm:memorysmall}, we can show the following lemma.

\begin{lemma}\label{lemma:np_easy}
    The constrained existence problem of XRSEs is in $\NP$. The same problem when players are restricted to pure strategies is still in $\NP$.
\end{lemma}

\begin{proof}
    Let $\Game_{\|v_0}$ be a simple quantitative stochastic game.
    Let $(P,O)$ be a partition of $\Pi$, and let $\bx$ and $\by$ be threshold vectors.
    By \cref{thm:memorysmall}, if there exists a (pure) XRSE with $\bx \leq \X(\bsigma) \leq \by$, then there exists one with at most $3np-2n+p+1$ memory states, where $p$ is the number of players and $n$ is the number of vertices.
    Such a strategy profile can be guessed in polynomial time.
    
    We now show that, once such a finite-memory strategy profile $\bsigma$ is guessed, one can check in polynomial time whether it is an XRSE, and satisfies the constraint $\bx \leq \X(\bsigma) \leq \by$.
    
    \begin{itemize}
        \item First, given $\bsigma$, for each player $i$, the quantity $\xr_i(\bsigma)$ can be computed in polynomial time, since it reduces to computing player $i$'s risk measure in the Markov chain induced by $\bsigma$ (which has polynomial size) (\cref{lm:secretlemma} in App.~\ref{appendix:secretlemma}).

        \item Second, checking that $\bx \leq \xr(\bsigma) \leq \by$ can be done in polynomial time.

        \item Third, for each player $i$, one must check that player $i$ has no profitable deviation.
        This can also be done in polynomial time (\cref{lm:secretlemma} in App.~\ref{appendix:secretlemma}) by computing the best risk measure player $i$ can get in the MDP induced by $\bsigma_{-i}$ (which has polynomial size).\qedhere \qedhere
    \end{itemize}
\end{proof}

\subsection{Restrictions on strategies}
We now consider subcases where the space of a strategies is restricted. 
We show in \cref{thm:infinite_rho_restricted_strategy_np_easy} that restricting the memory or amount of randomness of the strategy still renders the problem only in $\NP$.
Later in this section, we prove that all these problems, including the general problem, are $\NP$-hard. This subsection therefore completes the proof of \cref{thm:NPcomplete,thm:infinite_rho_restricted_strategy_np_easy}.

We restrict the set of strategies of each player to stationary, positional or pure. 
We show that the problem is in $\NP$ for each of these cases.
\begin{lemma}\label{lm:restrictionsNPeasy}
    The constrained existence problem, when all the players are restricted to positional, stationary, or pure strategies, is in $\NP$. 
\end{lemma}
\begin{proof}
    We show that we can still guess a strategy profile, and verify in polynomial time if it is indeed an XRSE.
    For the cases of positional and stationary strategies, guessing a strategy profile is straightforward, since such a strategy profile $\bsigma$ can be represented using polynomially many bits.
    We can  then verify that a given strategy profile $\bsigma$ gives risk measures within the constraints, and also is an XRSE in polynomial time (\cref{lm:secretlemma} in \cref{appendix:secretlemma}). 

    However, for pure strategies, memory might be required. But we showed with \cref{thm:memorysmall} that if there is a pure strategy profile, then there is one that requires polynomial memory, and therefore our results follow.  
\end{proof}
We now prove $\NP$-hardness of the constrained existence problem for the general setting as well as the cases where the players are restricted. 

\begin{restatable}[App.~\ref{app:np_hardness}]{lemma}{NPHard}\label{lemma:np_hardness}
    The constrained existence problem of XRSEs
    is $\NP$-hard, even when all players are pessimists and all rewards are non-negative.
    It remains $\NP$-hard when the strategies are reduced to stationary, pure, or positional ones.
\end{restatable}

\begin{proof}[Proof sketch]
     We reduce instances of $\THREESAT$ to a an instance of the problem. From a given formula $\Phi$, we construct  a game $\Game_\Phi$
    with no optimist and $4n+m+1$ pessimists, where $n$ is the number of literals and $m$ the number of clauses in $\Phi$.
    That game will contain an XRSE where a witness player gets risk measure $2$ if and only if $\Phi$ is satisfiable.
\end{proof}
 

This lemma, along with \cref{lemma:np_easy}, proves \cref{thm:NPcomplete}; and along with \cref{lm:restrictionsNPeasy}, it proves \cref{thm:infinite_rho_restricted_strategy_np_easy}.

\subsection{Things get easier when everyone is optimistic}
Since our $\NP$-hardness results involved only pessimistic players, we now show that
the constrained existence problem of XRSEs becomes $\PTIME$-complete when the perceived reward of each player is computed based on the risk measure $\oexp$, thus proving \cref{thm:PTIMEcomplete}.
We first show an upperbound by giving a polynomial-time algorithm.  

\begin{restatable}[App.~\ref{app:ptimeupperbound}]{lemma}{ptimeupperbound}\label{lm:ptimeupperbound}
    If all players are optimists, then the constrained existence problem for XRSE is in $\PTIME$, and there is an algorithm for the decision problem, which runs in time $\Oh(pm^2)$, where $m$ is the number of edges in $\Game$ and $p$ the number of players.
    Moreover, the algorithm can be modified to output an XRSE that satisfies the constraints, if one exists in time $\Oh(pm^2 + m^3)$.  Moreover, there is an algorithm that runs in time $\Oh(pm^2)$ if the upper bounds $y_i \geq 0$ for all players $i$.
\end{restatable}

\begin{proof}[Proof Sketch.]
    We want to decide whether there exists an XRSE $\bsigma$ satisfying the constraints $\Bar{x}\leq \X(\bsigma) \leq \Bar{y}$.
    The algorithm considers and deals with two cases, that we call \emph{cycle-friendly} and \emph{cycle-averse} cases, separately.
    In the cycle-friendly case, we have $y_i \geq 0$ for all players~$i$.
    Then, an XRSE could have positive probability of reaching no terminal vertex.
    However, in the cycle-averse case, that is impossible, since there is a player $i$ such that $y_i < 0$.
    In this proof sketch, we describe only the algorithm in the cycle-friendly case.

    The algorithm constructs a decreasing sequence of sets of edges $E_0, E_1, \dots$ until it reaches a fixed point.
    For each set $E_k$, it considers the strategy profile $\bsigma^{E_k}$, defined as follows: from each non-stochastic vertex $v$, when $v$ is seen for the first time, it randomises uniformly between all edges $vw \in E_k$.
    Later, if $v$ is visited again, it always repeats the same choice.
    If some player $i$ deviates and takes an edge that they are not supposed to take, then all the players switch to a positional strategy profile designed to minimise their risk measure.
    Such a strategy profile is finite-memory, but requires $2^{|V|}|V| + p$ memory states to be represented as a memory structure: we therefore use the set $E_k$ as a succinct representation.

    At each iteration $k$, the algorithm identifies new sets of vertices $V_\bad^k$ that must be avoided. This includes the terminals that give some player $i$ a payoff that is larger than $y_i$, or vertices from which player $i$ can deviate and obtain a higher value than the value offered by the strategy profile $\X_i(\bsigma^{E_k})$. 
    If it is not possible to avoid reaching the set $V_\bad^k$, the answer $\No$ is returned.
    Otherwise, the set $E_{k+1}$ is defined from $E_k$ by removing 
    edges that ensure that $V_\bad^k$ is not reached with positive probability.
    The algorithm stops when there are no more edges to remove and answers $\Yes$ and if we have $\X_i(\bsigma^{E_k}) \geq x_i$ for each $i$, and $\No$ otherwise.

    Each iteration requires time $\Oh(mp)$ to identify and remove edges.
    Since there are $\Oh(m)$ many edges, the algorithm terminates in time $\Oh(pm^2)$.
\end{proof}

Finally, we show that the problem is $\PTIME$-hard, even when there are only two players.

\begin{restatable}[App.~\ref{app:ptimelowerbound}]{lemma}{ptimelowerbound}\label{lm:ptimelowerbound}
    The constrained existence problem of XRSEs with optimistic players is $\PTIME$-hard even with only two players.
\end{restatable}

\begin{proof}[Proof sketch]
    We give a log-space reduction from the problem of deciding two-player zero-sum reachability games, which is known to be $\PTIME$-complete~\cite[Proposition~6]{Imm81}.
\end{proof}

\section{Discussion}
Our definition opens up several promising directions for future research.
One immediate extension of our work would be to study games with more sophisticated objectives, such as mean payoff or discounted sum. 
Another extension of our work is to study the concurrent version of such games, where players choose actions concurrently rather than in a turn-based setting. Concurrent stochastic multi-player games on stochastic graphs have traditionally been viewed as intractable, often requiring infinite memory to achieve optimal strategies, which has limited exploration of multi-player versions of the same problem. 

Finally, our definition of risk-sensitive equilibria is modelled after Nash equilibria and suffers from several of their limitations.
Like Nash equilibria, RSEs allow irrational behaviours when one player deviates and must be punished, 
as exemplified in our game \cref{fig:ex_extreme3}.
Exploring alternative definitions of risk-sensitive equilibria that are modelled after other equilibria concepts more suited for games on graphs~\cite[Section~7.1]{Osb04}, such as subgame-perfect equilibria, could provide a more rational framework for player decision-making.

\bibliography{biblio}

\appendix
\section{Appendix for \cref{sec:ERM}}\label{appendix:ERM}

\subsection{Proof of \cref{lemma:RSEtoQSSG}}\label{app:RSEtoQSSG}

\RSEtoQSSG*
\begin{proof}[Proof of \cref{lemma:RSEtoQSSG}]
    Consider the simple stochastic game $\Game_{\|v_0} = \tpl{V,E,\Pi,(V_i)_{i\in \Pi},\p ,\mu}$. We will define a payoff  function $\mu'$ over the same set of terminals for game $\Game'_{v_0}$ such that $\Game'_{\|v_0} = \tpl{V,E,\Pi,(V_i)_{i\in \Pi},\p ,\mu'}$ has a Nash equilibrium if and only if $\Game$ has a $(\beta,\brho)$-ERSE.

    For a terminal vertex $t$, we simply define $\mu_i'(t) = 1-(\beta^{-\rho_i\mu_i(t)})$ if $\rho>0$, and $\mu_i'(t) = (\beta^{-\rho_i\mu_i(t)})-1$ if $\rho<0$.
    
    Consider the function $\modifiedreward{\beta}{\rho}\colon x\mapsto 1-(\beta^{-\rho x})$ if 
    $\rho>0$ and $x\mapsto (\beta^{-\rho x})-1$ if $\rho<0$
    as the modified reward function.  This function is similar to the negative utility function defined in the work of Baier et al.,~\cite{BCMP24}, where they replace terminal rewards with the negative value of $(\beta^{-\rho\mu_i(t)})$ (as they assume $\rho>0$), in order to compute the winner in a two-player zero-sum game with risk-averse players. We additionally add or subtract $1$ from their value  to ensure that besides monotonicity, this function also maps the play that does not reach a terminal in the original game to the payoff $0$, and therefore in the modified game to preserve that such plays are still mapped to $0$. 
    
    Observe that for any random variable $X$ and constant $r$, for a value $\rho>0$, we have that 
    \begin{equation}\label{inequality:REvsExp}    
    \re_{\beta,\rho}\left[X\right] \geq r\text{ if and only if }\Eb\left[\modifiedreward{\beta}{\rho}(X)\right] \geq 1-\beta^{-\rho r}\end{equation}
    since $\re_{\beta,\rho}\left[X\right] = \frac{-1}{\rho}\log_\beta\tpl{\Eb[\beta^{-\rho X}]}$.
    
    Therefore, any Nash equilibrium in the game $\Game'_{\|v_0}$ implies that there is a strategy profile $\bsigma$ such that, for all players $i\in\Pi$, in the MDP induced by $\bsigma_{-i}$, the strategy $\sigma_i$ of player $i$ is an optimal strategy. 

    We first consider the case of player $i$ where $\rho_i>0$. The case of $\rho_i<0$ is analogous, so we omit it. 
    For every strategy $\tau_i$ of player $i$, where $\rho_i>0$, if we write $\btau = (\bsigma_{-i}, \tau_i)$, we have 
    $\Eb(\bsigma)[\mu_i']\geq \Eb(\btau)[\mu_i']$, since $\bsigma$ is a Nash equilibrium. 
    Since the payoffs of $\Gc'$ at a terminal $t$ is just $\modifiedreward{\beta}{\rho}(\mu_i(v))$, we therefore have $\Eb(\bsigma)[\modifiedreward{\beta}{\rho_i}(\mu_i)]\geq \Eb(\btau)[\modifiedreward{\beta}{\rho_i}(\mu_i)]$. From  \cref{inequality:REvsExp}, we have 
    \begin{align*}
         \Eb(\bsigma)\left[\modifiedreward{\beta}{\rho_i}(\mu_i)\right]\geq \Eb(\btau)\left[\modifiedreward{\beta}{\rho_i}(\mu_i)\right]\text{ if and only if }\\
          \Eb(\bsigma)\left[1-\beta^{-\rho_i \mu_i}\right]\geq \Eb(\btau)\left[1-\beta^{\rho_i \mu_i}\right] \\
           \iff\Eb(\bsigma)\left[-\beta^{-\rho_i \mu_i}\right]\geq \Eb(\btau)\left[-\beta^{-\rho_i \mu_i}\right]
    \end{align*}
Taking $\frac{1}{\rho}\log_\beta$ on both sides,   we get the above is true iff
    \begin{align*}
         \re_{\beta,\rho_i}(\bsigma)[\mu_i] \geq -\frac{1}{\rho_i}\log_\beta\tpl{-\Eb\left[\modifiedreward{\beta}{\rho_i}(\btau)[\mu_i]\right] }\\ \text{if and only if }\re_{\beta,\rho_i}[\bsigma](\mu_i) \geq \re_{\beta,\rho_i}[\btau](\mu_i)     \end{align*}
      Therefore the strategy  $\bsigma$ is at least as good as (any strategy where one player deviates) $\btau$ for the player $i$ that deviates, when their rewards are the risk-entropy measure. Thus, the strategy profile $\bsigma$ is an ERSE.

\end{proof}

\subsection{Proof of \cref{thm:ERRSErestricted}}\label{app:ERRSErestricted}

\stationaryRSE*
\begin{proof}[Proof of \cref{thm:ERRSErestricted}]
For the proof of \cref{thm:ERRSErestricted}, we first focus on \cref{itm:ERRSEitmundec}.
\begin{proposition}
    The constrained existence problem of $(\beta,\brho)$-ERSEs in quantitative stochastic games where players are restricted to pure strategies is undecidable.
\end{proposition}
\begin{claimproof}
The undecidability of the case where pure strategies are considered is inherited from Nash equilibria~\cite[Theorem~4.9]{UW11}, since the reduction for undecidability uses only pure strategies. Therefore, our current undecidability follows from \cref{proposition:Undecidable,lemma:RSEtoQSSG}.  
\end{claimproof}

\paragraph*{Decidable subcases.} Now, we turn our attention to \cref{itm:ERRSEdecidable} to show decidability and conditional decidability results. 
\subparagraph*{Decidability lowerbounds.} The problem is $\NP$-hard in general for even two players, which follows from  the work of Ummels and Wojtczak.
Since Nash equilibria are a specific instance of the setting of ERSEs, where the risk parameters of each player is $0$, $\NP$-hardness follows~\cite[Theorem 4.4]{UW11}. 
Similarly, for stationary strategy profiles, they show $\SQRTSUM$-hardness~\cite[Theorem 4.6]{UW11} for Nash equilibria which also shows a lower bound for our case. 

\subparagraph*{Decidability upperbounds.} For decidability, we show in \cref{prop:ETRformula} that we can write a formula in the existential theory of reals (with exponentiation if $\beta=e$)
which is satisfied if and only if the constrained existence problem is satisfied. This gives the $\PSPACE$ upper bound when $\beta$ is algebraic since the existential theory of reals is in $\PSPACE$~\cite{Can88}, and decidability subject to Shanuel's conjecture when the base is $\beta = e$, since the existential theory of reals with exponentiation is decidable assuming Shanuel's conjecture~\cite{MW95}.
Our proof is similar to the one by Ummels and Wojtczak~\cite[Theorem 4.5]{UW11}, however, we need to do slightly more work to encode the payoff expressed by the entropic risk measure. 

We only have to show that we can encode the constrained existence problem using the existential theory of reals. 
First, observe that it is enough to verify if there is a memoryless Nash equilibrium in the modified game obtained where all the terminal rewards $\mu_i(v)$ are replaced instead with $1-\tpl{\beta^{\rho\mu_i(v)}}$. This follows from \cref{lemma:RSEtoQSSG}. 

Since the players are restricted to strategies that are memoryless, we give a non-deterministic algorithm that uses the solution to sentences in $\exists\Rb$ if the values of $\beta$ and $\rho$s can be expressed in $\Qb$.  Since $\NPSPACE = \PSPACE$, and a non-deterministic procedure is still a deterministic procedure, this does not change the complexity. 
     
     For a game $\Game_{\|v_0} = \tpl{V,E,\Pi,(v_i)_{i\in \Pi}, \p,\mu}$, where $\mu$ is the payoff function from a terminal set of nodes $T$ to values in $\Qb$, and constraints $\Bar{x}$ and $\Bar{y}$ for each of the $n$ players in $\Pi$, our algorithm guesses, first, the support $S \subseteq E$ of the strategies that will be considered; that is, the set of edges that will be used with positive probability. 

\begin{claim}
            For any $z$, which requires $\ell$ bits to encode, there is a formula in $\exists\Rb$ that uses only polynomially many variables in $\ell$ to encode $\modifiedreward{\beta}{\rho}(z) = 1-\beta^{-\rho z}$, where $\beta$ and $\rho$ can also be represented in $\exists \Rb$ using a polynomial formula.  If $\beta = e$, then $\modifiedreward{e}{\rho}(z) = 1-\beta^{-\rho z}$ can be expressed using the existential theory of reals with exponentiation using a formula of most polynomial length.
\end{claim}
\begin{claimproof}
            We assume without loss of generality that $\beta$ is a natural number. If $\beta$ is rational instead, and is represented by a value $a/b$, then individually find $a_1  = a^{\rho z}$ and $b_1 = b^{\rho z}$, and just find $b_1/a_1$, which is just written in $\exists\Rb$ by stating that $\exists  t_r \;\exists a_1'\colon a_1'\times a_1  = 1\land t_r = a_1'\times b_1$.
        Now that we assume that $\beta$ is a natural number, we deal with fixed finite exponentiation with rational values. Similarly, we can assume without loss of generality that $\rho z$ is a natural number. Indeed, any value $r^{a/b} = r^a\times z^{1/b}$ and $z^{1/b}$ can be written as $\exists y \colon y^b = z$.

        It suffices to show therefore that for two values $b,a$, both natural numbers, $b^a$ can be expressed in $\exists\Rb$ succinctly, using only a formula that has length that is not more a poly-log of $b$ or $a$. 
        Let $a = \sum_{i=0}^{\log_2{a}}a_i 2^i$, where $a_i\in\{0,1\}$. 
        This follows from the following observations. 
        \begin{itemize}
            \item $b^a = \prod_{i=1}^{\log_2{a}}\tpl{b^{a_i}b^{2^i}}$
            \item  $2^i$ can be expressed in a formula with at most $i+1$ many variables
            \item $b^{2^i}$ further requires at most $i$ many variables to express, because if $b_i$ represents $b^{2^{i}}$, then we have $b^{2^{i+1}} = b^{2^i}\times b^{2^i}$.
            \item Finally, using a similar trick, $b$ itself can be represented using at most $\log_2{b}$-many variables. 
        \end{itemize}
        If $\beta = e$, it naturally follows that $e^{-\rho z}$ is expressed using exponentiation with $e$. 
\end{claimproof}

    The above claim ensures we can efficiently represent the variables used for the payoffs of the modified game, we now can write an equation assuming that all terminal rewards are available to us as constants. 
    This will write the equation in three parts. Since we have guessed the support, we first ensure that, in fact, there are variables corresponding to the probabilities of the strategy that only take positive values on the edges corresponding to the support set that we guessed.  
    Then, we write equations using variables that compute the values of the induced Markov chain from this strategies. Finally, we also have a formula whose solution corresponds to the values of the MDP obtained for each player when playing against the strategies of all other players. Then we compare if the value of the MDP is at least as large as the underlying Markov chain for each player, to ensure that it is indeed an equilibrium. To write all of this in $\exists\Rb$, we introduce the following variables. 
     \begin{itemize}
         \item one variable  $p_{vw}$ for each pair of vertices $vw$, which corresponds to the probabilities corresponding to the strategy;
         \item a variable $r^i_v$ which corresponds to the entropic risk measure of player $i$ from vertex $v$ if they followed the strategy defined by the probabilities above; 
         \item a variable $m^i_v$ which corresponds to the value obtained by player $i$ if the game is treated as an MDP against other players.
     \end{itemize}
 
     \begin{proposition}\label{prop:ETRformula}
        For the constrained existance of ERSE in a game $\Game_{\|v_0}$ with constraints $\Bar{x},\Bar{y}$, and a subset $S$ of edges of the game, 
         \begin{itemize}
             \item there is a formula  in $\exists\Rb$ that is satisfied if and only if there is a stationary strategy that uses exactly edges in $S$, when $\beta$ and $\rho_i$ are rational values.
             \item there is a formula $\Gamma_S'$ in $\exists\Rb\text{-}\exp$ that is satisfied if and only if there is a stationary strategy that uses exactly the edges in $S$ when $\beta = e$ and $\rho_i$ are rational values.
         \end{itemize}
    \end{proposition}
    \begin{claimproof}
    The following part of the proof is similar to the one found in Ummels and Wojtczak~\cite[Theorem 4.5]{UW11}, but we provide it to suit our setting, for the sake of completeness. 
    First, we have a formula that states that the values $p_{vw}$ indeed describe a strategy. We further ensure that for stochastic vertices, the value $p_{vw}$ encodes exactly the value dictated by the probability function $\p$ by the stochastic vertex:
    \begin{align*}
    \Phi_S(\Bar{p})\:= \bigwedge_{v,w\in V} \left( p_{vw}\geq 0\right)\land \bigwedge_{v,w\in V} \left( p_{vw} \leq 1\right)
    \land \bigwedge_{i\in \Pi}\bigwedge_{v\in V_i} \tpl{\sum_{w\in E(v)}\p_{vw}=1}\land \\
    \bigwedge_{v\in V_?} \left( p_{vw} =  \p(vw)\right) \land \bigwedge_{vw\in S} \left(p_{vw}> 0 \right)
    \end{align*}
    For a fixed support $S$ of a strategy $\bsigma$, it is possible to compute the terminals $T_S$ that have non-zero probability of being reached in the underlying Markov chain that is formed, and the vertices $V_S$ from which such terminals can be reached with non-zero probability. We assign value $\mu_i'(v)$ as the reward for the terminal vertices for player $i$, and the reward $0$ for all vertices that cannot reach any terminal with positive probability. 
    \[\Omega_S^i(\Bar{p},\Bar{r}^i)\:= \bigwedge_{t\in T_S} \left(r^i_t =\mu_i'(t)\right) \land
                 \bigwedge_{v\notin V_S} \left( r^i_v = 1\right) \land
                 \bigwedge_{v\in V_S\setminus T_S} \tpl{r^i_v = \sum_{w\in E(v)}p_{vw} r^i_v}
                 \]     
    Finally, for computing the values of the MDP, we construct a similar FO statement
    \[\Psi_S^i(\Bar{p},\Bar{m}^i) \:= \bigwedge_{t\in T} \left( m^i =\mu_i'(t) \right) \land
                 \bigwedge_{v\in V_i, w\in E(v)} \left(m_v^i\geq m_w^i\right) \land
                 \bigwedge_{v\notin V\setminus V_i} \tpl{m^i_w = \sum_{w\in E(v)} p_{vw} m^i_v}\]
Finally our statement would be $$\exists \Bar{p}\:\exists\Bar{r}\:\exists\Bar{m}\colon \Phi(\Bar{p})\land \bigwedge_{i\in\Pi}\left(\tpl{x^i_{v_0}\leq r^i_{v_0}}\land \tpl{r^i_{v_0}\leq y^i_{v_0}}\land\Omega_S^i(\Bar{p},\Bar{r^i})\land \Psi_S^i(\Bar{p},\Bar{m^i})\land \left( m_{v_0}^i\leq r_{v_0}^i \right) \right)$$
For rewards that are represented by algebraic numbers that also can be expressed succinctly via $\exists\Rb$, we observe that our above reduction extends naturally. 
For $\beta = e$, we remark that the same formula is expressible using $\exists\Rb\text{-}\exp$.
    \end{claimproof}

\end{proof}

\section{A technical lemma}\label{appendix:secretlemma}
\begin{lemma}\label{lm:secretlemma}
    Let $\Game_{\|v_0}$ be a game with two players, called $i$ and $j$.
    We assume given a partition $(P, O)$ of $\{i, j\}$.
    Then, the quantity:
    $$\inf_{\sigma_j \in \Strat_j\Game_{\|v_0}} \sup_{\sigma_i \in \Strat_i\Game_{\|v_0}} \X_i(\sigma_i, \sigma_j)$$
    can be computed in time $\Oh(m)$, where $m$ is the number of edges in $\Game$.
    Moreover, the infimum is reached with a positional strategy of player $j$; and there is a positional strategy of player $i$ that realises the supremum for every strategy of player $j$.

    Consequently, the optimality of positional strategies and the $\Oh(m)$ upper bound also hold in Markov decision processes, and in Markov chains; and, on the other hand, it holds when $j$ is a fictional player that represents a coalition of players who all have as unique objective to minimise player $i$'s risk measure.
\end{lemma}

\begin{proof}
    The $\Oh(m)$ upper bound holds by a slight adaptation of the classical attractor algorithm~\cite[Chapter~5.3]{AG11}.
    Note that those algorithms run in time $\Oh(m+n)$, where $n$ is the number of vertices; but here, we assumed that each vertex (except possibly $v_0$) has at least one ingoing edge, hence $n \leq m+1$ and $m+n = \Oh(m)$.
    That algorithm immediately induces positional optimal strategies.
    Another way to obtain that second result, however, is the following: once the quantity $x = \inf_{\sigma_j} \sup_{\sigma_i} \X_i(\sigma_i, \sigma_j)$ is known, strategies that realise the infimum and the supremum can be seen as optimal strategies in the Boolean zero-sum game in which player $i$ wants with positive probability (if they are optimist) or with probability $1$ (if they are pessimist) to reach the set of terminals yielding them at least payoff $x$ (if $x > 0$) or to avoid the set of terminals yielding them less than payoff $x$ (if $x \leq 0$).
    This is then a reachability game (seen either from player $i$'s of from player $j$'s perspective), and it is well-known~\cite[Chapter~5.3]{AG011} that in such a game, for both players, positional strategies suffice to maximise the probability of winning.
    In particular, if one has a strategy to win that game with positive probability, or with probability $1$, there is also such a strategy that is positional.
\end{proof}
\section{Appendix for \cref{sec:XR}}\label{appendix:XR}
\subsection{Proof of \cref{thm:RE=PEorOE}}\label{app:RE=PEorOE}

\REisPEOE*
\begin{proof}[Proof of \cref{thm:RE=PEorOE}]
    \begin{itemize}

        \item First, let us note that for every $\rho$, we always have $\re_{\beta,\rho}[X] \geq \pexp[X]$.
        Let now $\epsilon > 0$.
        We want to prove that there exists $\rho_0 \in \Rb$ such that for every $\rho \geq \rho_0$, we have $\re_{\beta,\rho}[X] \leq \pexp[X] + \epsilon$.

        Let us first notice that we have:
       \begin{align*}
       \re_{\beta\rho}[X] &= -\frac{1}{\rho} \log_\beta \left( \int_{x \in \Rb} \beta^{-\rho x} \d \prob(X = x) \right)\\
        &= -\frac{1}{\rho} \log_\beta \left( \int_{x \in \Rb} \beta^{-\rho \pexp[X]} \beta^{-\rho (x-\pexp[X])} \d \prob(X = x) \right)\\
        &= \pexp[X] -\frac{1}{\rho} \log_\beta \left( \int_{x \in \Rb} \beta^{-\rho (x-\pexp[X])} \d \prob(X = x) \right)\\
        &= \pexp[X] -\frac{1}{\rho} \log_\beta \Bigg( \int_{x \leq \pexp[X] + \frac{\epsilon}{2}} \beta^{-\rho (x-\pexp[X])} \d \prob(X = x) \\
        &\qquad\qquad\qquad+ \int_{x \geq \pexp[X] + \frac{\epsilon}{2}} \beta^{-\rho (x-\pexp[X])} \d \prob(X = x) \Bigg)\\
         &\leq \pexp[X] - \frac{1}{\rho} \log_\beta \left( \int_{x \leq \pexp[X] + \frac{\epsilon}{2}} \beta^{-\rho \frac{\epsilon}{2}} \d \prob(X = x) + 0 \right)\\
        &= \pexp[X] - \frac{1}{\rho} \log_\beta \left( \prob\left(X \leq \pexp[X] + \frac{\epsilon}{2}\right) \beta^{-\rho \frac{\epsilon}{2}} \right)\\
        &= \pexp[X] - \frac{1}{\rho} \log_\beta \left( \prob\left(X \leq \pexp[X] + \frac{\epsilon}{2}\right)\right) + \frac{\epsilon}{2}.
        \end{align*}

        For $\rho$ large enough, this quantity is indeed smaller than $\pexp[X] + \epsilon$.

        \item Let us first notice that for every $\beta, \rho, X$, we have the following equality $\re_{\beta,\rho}[X] = -\re_{\beta(-\rho)}[-X]$.
        Thus, we can apply the previous result, and find:
        \begin{align*}
        \lim_{\rho \to -\infty} \re_{\beta,\rho} [X] 
        &= \lim_{\rho \to -\infty} -\re_{\beta(-\rho)} [-X]
        \\&= -\lim_{\rho \to +\infty} \re_{\beta,\rho} [-X]
        \\& =  - \pexp[-X]
        \\ &= - \inf \{x \in \Rb ~|~ \prob(-X \leq x) > 0\}
        \\ &= - \inf \{x \in \Rb ~|~ \prob(X \geq -x) > 0\}
        \\ &= \sup \{x \in \Rb ~|~ \prob(X \geq x) > 0\}
        \\ &= \oexp[X].
        \end{align*}
        \end{itemize}
\end{proof}





\subsection{Proof of \cref{thm:XRSEexists}}\label{app:XRSEexists}

\XRSEexists*

\begin{proof}[Proof of \cref{thm:XRSEexists}]
    Throughout this proof, for a given set of edges $F \subseteq E$, we write $\Game^F$ for the game obtained from $\Game$ by removing all the edges that do not belong to $F$.
    In that game, we define $\bsigma^F$ as the stationary strategy profile that maps each vertex $v$ to some probability distribution whose support is $F(v)$.
    Note that the probabilities do not matter here: we are only interested in the support of the distribution of the strategy profile.

    \paragraph*{Algorithm.} We proceed by presenting the algorithm, \cref{algo:existence}, that takes as an input the game $\Game_{\|v_0}$ and the partition $(P, O)$, and returns a subset $F \subseteq E$ such that, as we will show, the strategy profile $\bsigma^F$ is always an XRSE.
    That algorithm defines a decreasing sequence $E_0, E_1, \dots$ of subsets of $E$, where $E_0 = E$.
    At each step $k$, for each pessimist $i$, it computes the risk measure $z_i^k$ of player $i$ in $\bsigma^{E_k}$, and then the set $W_i^k$ of vertices $v$ such that, from $v$, whatever player $i$ does, that player almost surely gets a payoff smaller than or equal to $z_i^k$.
    If we have $v_0 \in W_i^k$ for each $i$, then the algorithm stops there and returns the set $E_k$ (and we will show below that it means that $\bsigma^{E_k}$ is an XRSE).
    Otherwise, we pick player $i$ such that $v_0 \not\in W_i^k$ (a player who provably has a profitable deviation), and define $E_{k+1}$ by removing all the edges accessible from $v_0$ leading from $V \setminus W_i^k$ to $W_i^k$.

            \begin{algorithm}
            \begin{algorithmic}\caption{Exhibition of one stationary XRSE}\label{algo:existence}
                \Procedure{Existence}{$\Game_{\|v_0}, P, O$}
                    \State $k \gets 0$
                    \State $E_k \gets E$
                    \While{$\top$}
                        \State Compute $A^k = \{v \in V \mid v \text{ is accessible from } v_0 \text{ in } (V, E_k)\}$
                        \ForAll{$i \in P$}
                            \State Compute $z^k_i = \X_i(\bsigma^{E_k})$
                            \State Compute $W^k_i = \{v \in V \mid \forall \tau_i \in \Strat_i \Game^{E_k}_{\|v_0}, \text{we have } \prob_{\bsigma^{E_k}_{-i}, \tau_i}(\mu_i \leq z^k_i) > 0\}$
                        \EndFor
                        \If{$\exists i$ such that $v_0 \not\in W_i^k$}
                            \State Pick one such $i$
                            \State $E_{k+1} \gets E_k \setminus ((A^k \setminus W_i^k) \times W_i^k)$
                            \State $k \gets k+1$
                        \Else
                            \State \Return $E_k$
                        \EndIf
                    \EndWhile
                \EndProcedure
            \end{algorithmic}
        \end{algorithm}

    \paragraph*{Correctness} A first quick invariant that we need to prove is the following one, which will guarantee that the games $\Game^{E_k}$ and the strategies $\bsigma^{E_k}$ are well-defined.

    \begin{invariant}\label{inv:outgoingedges}
        For each $k$, each stochastic vertex $v$, we have $E(v) \subseteq E_k$, and for each non-stochastic vertex $v$, we have $E(v) \cap E_k \neq \emptyset$.
    \end{invariant}
    
\begin{claimproof}[Proof that \cref{algo:existence} satisfies \cref{inv:outgoingedges}]
    The set $E_0 = E$ trivially satisfies the invariant.

    Now, let us assume that $E_k$ satisfies the invariant.
    At step $k$, an edge is removed if and only if goes from a vertex $u \in A^k \setminus W^k_i$ to a vertex $v \in W^k_i$.
    Consider a stochastic vertex $u \in A^k$: if it has an edge that leads to vertex $v \in W^k_i$, then whatever player $i$ plays from $u$, with positive probability, the vertex $v$ is reached; and then, if the other players play the strategy profile $\bsigma^{E_k}$, then with positive probability, player $i$ gets the payoff $z^k_i$ or less.
    Hence $u \in W_i^k$, and the edge $uv$ is not removed, and remains in the set $E_{k+1}$.
    Similarly, if $u$ is not a stochastic vertex, but all its outgoing edges lead to a vertex that belong to $W_i^k$, then $u$ itself belongs to $W_i^k$, hence the outgoing edges of $u$ will not all be removed.
    The invariant is therefore still true at step $k+1$, and by induction, is true for all $k$.
\end{claimproof}
    
Each step of the algorithm is then also properly defined.
Moreover, we have termination.

\begin{proposition}
    \cref{algo:existence} terminates.
\end{proposition}

\begin{claimproof}
    With \cref{inv:outgoingedges}, we now know that \cref{algo:existence} successfully constructs a sequence $E_0, E_1, \dots$ of sets of edges until it stops and returns the last of those sets.
    Termination is an immediate consequence of the fact that this sequence is decreasing.

    Indeed, for each step $k$ at which nothing is returned, there exists a player $i$ with $v_0 \not\in W_i^k$.
    On the other hand, the set $W_i^k$ is necessarily accessible from $v_0$:

    \begin{claim}
        The set $W_i^k$ is nonempty, and accessible from $v_0$ in the graph $(V, E_k)$.
    \end{claim}

    \begin{claimproof}
        If $z^k_i$ is obtained by reaching a terminal vertex $t$, then we have $t \in W_i^k$, and $t$ is accessible from $v_0$.
        If now $z^k_i = 0$ is obtained by reaching no terminal vertex, then when following $\bsigma^k$, with positive probability, no terminal is reached.
        Then, there is in particular a vertex $u$ that has positive probability to be visited infinitely often.
        And when playing $\bsigma^k$ from $u$, the probability that some terminal is ever reached is actually $0$, since if it was some constant $q > 0$, then the probability of visiting $u$ infinitely often would be $\lim_\l (1-q)^\l = 0$.
        In other words, no terminal vertex is accessible from $u$ in $(V, E_k)$, and then, we have $u \in W_i^k$.
    \end{claimproof}

    Now, along a play that starts from $v_0 \not\in W_i^k$ and visits $W_i^k$, there exists at least one edge that goes from a vertex that does not belong to $W_i^k$, to a vertex that does.
    Such an edge is then removed in the set $E_{k+1}$, which is therefore strictly included in the set $E_k$.
    This holds for every $k$, ensuring termination.
\end{claimproof}

We now know that the algorithm terminates, i.e., constructs a finite decreasing sequence $E = E_0, E_1, \dots, E_n$, and then returns the set $E_n$, as a succinct representation of the stationary strategy profile $\bsigma^{E_n}$.
What remains to be proven is that that strategy profile is an XRSE.
Before proving that it is an XRSE in the game $\Game_{\|v_0}$, we first prove that it is one in the game $\Game^{E_n}_{\|v_0}$, i.e., when the edges that have been removed cannot be used to deviate.

\begin{proposition}\label{prop:xrseGEn}
    The strategy profile $\bsigma^{E_n}$ is an XRSE in the game $\Game^{E_n}_{\|v_0}$.
\end{proposition}

\begin{claimproof}
    Consider a player $i$, and a deviation $\sigma'_i$ of player $i$ from the strategy profile $\bsigma^{E_n}$ in the game $\Game_{\|v_0}^{E_n}$.
    Let $x = \X_i(\bsigma^{E_n}_{-i}, \sigma'_i)$.

    \subparagraph*{If player $i$ is an optimist.}
    If $x = 0$, then since all rewards are non-negative, we have $x \leq \X_i(\bsigma^{E_n})$.
    If $x > 0$, then the payoff $x$ is obtained by reaching a terminal vertex $t$.
    But then, that terminal vertex is accessible from $v_0$ in the graph $(V, E_n)$, and is therefore also reached with positive probability when all players follow the strategy profile $\bsigma^{E_n}$.
    Hence, again, the inequality $x \leq \X_i(\bsigma^{E_n})$.
    
    \subparagraph*{If player $i$ is a pessimist.}
    Then, since the algorithm terminated at step $n$, player $i$ is such that $v_0 \in W_i^k$.
    The strategy $\sigma'_i$, like every strategy $\tau_i$ for player $i$, satisfies therefore the inequality $\prob_{\bsigma^{E_n}_{-i}, \sigma'_i}(\mu_i \leq z_i^n) > 0$.
    Consequently, we have $x \leq z_i^k = \X_i(\bsigma^{E_n})$.

    In both cases, the deviation $\sigma'_i$ is not profitable, hence the conclusion.
\end{claimproof}

Let us now prove that putting back the removed edges does not change that result, and therefore conclude the correctness proof.

\begin{proposition}
    The strategy profile $\bsigma^{E_n}$ is an XRSE in the game $\Game_{\|v_0}$.
\end{proposition}

\begin{claimproof}
    Let $i$ be a player, and let $\sigma'_i$ be a deviation from $\bsigma^{E_n}$ for player $i$ in $\Game_{\|v_0}$.
    Since $\bsigma^n$ is stationary, we can assume that $\sigma'_i$ is positional by \cref{lm:secretlemma}.
    If the strategy $\sigma'_i$ uses only edges of $E_n$, then it can be considered as a deviation from $\bsigma^{E_n}$ in the game $\Game^{E_n}$, hence by \cref{prop:xrseGEn}, it is not a profitable deviation.
    
    Let us now assume that the strategy $\sigma'_i$ uses an edge that does not belong to $E_n$, i.e. there exists a vertex $v$ that is visited with positive probability in the strategy profile $(\bsigma^{E_n}_{-i}, \sigma'_i)$ and an edge $vw \in E \setminus E_n$ such that $w = \sigma'_i(v)$.
    Since only edges controlled by pessimists have been removed, we can immediately deduce that player $i$ is a pessimist.
    
    Now, among such edges, let us choose one whose removal occurred the earliest, that is, let us choose it in order to minimise the index $k$ such that $vw \in E_k \setminus E_{k+1}$.
    Thus, in the strategy profile $(\bsigma^{E_n}_{-i}, \sigma'_i)$, it is almost sure that only edges of $E_k$ are used.

    The fact that the edge $uv$ has been removed at step $k$ means that we had $u \not\in W_i^k$ and $v \in W_i^k$.
    Thus, from the vertex $v$, if player $i$ uses only edges of $E_k$ (which is the case when they follow $\sigma'_i$), and if the other players follow the strategy profile $\bsigma^{E_k}_{-i}$, player $i$ gets the payoff $z_i^k$ or less with positive probability.
    Let us show that it is also the case when the other players follow the strategy profile $\bsigma^{E_n}$ instead of $\bsigma^{E_k}$.

\begin{claim}
    From the vertex $v$, we have $\prob_{\bsigma^{E_n}_{-i}, \sigma'_i}(\mu_i \leq z_i^k) > 0$.
\end{claim}

\begin{claimproof}
    We proceed by proving that when the players follow, from the vertex $v$, the strategy profile $(\bsigma^{E_k}, \sigma'_i)$, there is a positive probability that player $i$ gets the payoff $z_i^k$ or less \emph{and} that the set $W_i^k$ is never left.

    Indeed, if in that strategy profile there is a positive probability that player $i$ gets a payoff smaller than or equal to $z_i^k$ by reaching a terminal vertex $t$ that yields such a payoff, then we have $t \in W_i^k$, and with positive probability the terminal vertex $t$ is reached without leaving $W_i^k$.
    Similarly, if such a payoff is obtained by reaching no terminal vertex, and therefore getting payoff $0$, then, using a reasoning that has already been used above, with positive probability a vertex $w$ is reached without leaving $W_i^k$, such that from $w$, no terminal vertex is accessible anymore in $(V, E_k)$; and, therefore, all the vertices accessible from $w$ in that graph belong to $W_i^k$, hence once $w$ is reached it is almost sure that $W_i^k$ is never left.

    Then, the set $E_{k+1}$ was defined so that $W_i^k$ is no longer accessible from $v_0$ in the graph $(V, E_{k+1})$.
    Therefore, those vertices are not accessible at any step $\l > k$, and therefore no outgoing edge of a vertex of $W_i^k$ is ever removed in the sequel, i.e. $E_n \cap (W_i^k \times V) = E_k \cap (W_i^k \times V)$.
    Consequently, since $\sigma'_i$ uses only edges of $E_k$, when the strategy profile $(\bsigma^n, \sigma'_i)$ is played from $v$, it is also true that with positive probability player $i$ gets the payoff $z_i^k$, or less.
\end{claimproof}
    
    This claim proves that we have $\X_i(\bsigma^{E_n}_{-i}, \sigma'_i) \leq z_i^k$.
    To conclude that the deviation $\bsigma'_i$ is not profitable, we still need to prove that that quantity $z_i^k$ is smaller than or equal to (actually strictly smaller) the risk measure $\X_i(\bsigma^n)$.
    That inequality is an immediate consequence of the following claim.

\begin{claim}
    For every pessimistic player $j$, the sequence $(z_j^\l)$ of player $j$'s risk measures is non-decreasing.
\end{claim}

\begin{claimproof}
    Let $\l$ be a step index, and let us prove that we have $z_j^\l < z_j^{\l+1}$.
    The quantity $z_j^{\l+1}$ is the pessimistic risk measure of player $j$ in the strategy profile $\bsigma^{E_{\l+1}}$: there is therefore a positive probability that player $j$ gets the payoff $z_j^{\l+1}$ when that strategy profile is followed.
    Player $j$ obtains that payoff either by reaching a terminal vertex to which that payoff is assigned, or by reaching no terminal vertex at all.

    Let us first show that the second case is actually impossible.
    If player $j$ gets the payoff $z_j^{\l+1} = 0$ by reaching no terminal vertex, with the same reasoning as above, there is, in particular, a vertex $v$ that has positive probability to be visited infinitely often when $\bsigma^{\l+1}$ is played from $v_0$, and therefore such that no terminal vertex is accessible from $v$ in $(V, E_{\l+1})$.
    But then, let $j'$ be the player that was controlling the edges that were removed at step $\l$.
    Let us consider a strategy $\tau_{j'}$ of player $j'$ that uses only edges of $E_\l$.
    Then, when the strategy profile $(\bsigma^\l_{-j'}, \tau_{j'})$ is played from the vertex $v$, it will almost surely be true that either no terminal vertex is reached, leading to the payoff $0$, or an edge of $E_\l \setminus E_{\l+1}$ is taken, leading therefore to a vertex of $W_i^\l$, and to a risk measure of $z_{j'}^\l$ or less.
    Thus, since all rewards are non-negative and therefore $z_{j'}^\l \geq 0$, the vertex $v$ belongs to the set $W_{j'}^\l$, which is impossible since it should then have been made unaccessible in the graph $(V, E_{\l+1})$.

    Therefore, player $j$ gets payoff $z_j^{\l+1}$ by reaching a terminal giving them that payoff, which means that such a terminal is accessible from $v_0$ in the graph $(V, E_{\l+1})$.
    Then, it is also accessible from $v_0$ in the graph $(V, E_\l)$, and therefore it is reached with positive probability when following the strategy profile $\bsigma^{E_\l}$.
    Consequently, we have $z_j^\l \leq z_j^{\l+1}$.
\end{claimproof}

    Consequently, with $j = i$, we have $z_i^k \leq z_i^n$, and therefore $\X_i(\bsigma^{E_n}_{-i}, \sigma'_i) \leq z_i^k \leq z_i^n = \X_i(\bsigma^{E_n})$.
    The strategy $\sigma'_i$ is not a profitable deviation from the strategy profile $\bsigma^{E_n}$, which is therefore a (stationary) XRSE.
\end{claimproof}

    \paragraph*{Complexity}

    We finally show that \cref{algo:existence} runs with time $\Oh(m^2p)$.
    At each iteration of the while loop, at least one edge is removed; we therefore have at most $m$ iterations of the while loop.

    Now, during the $k^\text{th}$ iteration of the loop, the algorithm computes the set $A^k$, and for each pessimistic player $i$, the algorithm also computes the quantity $z_i^k = \X_i(\bsigma^{E_k})$, and then the set $W^k_i$.
    All of those computations can be done in time $\Oh(m)$ using \cref{lm:secretlemma}. Since there are $p$ players, this step therefore takes $\Oh(mp)$ time.
    Finally, checking whether $v_0 \in W^k_i$ for each player $i$ takes time $\Oh(p)$, and removing all the edges leading to $W_i^k$ to define the set $E_{k+1}$ takes time $\Oh(m)$.
    Hence, the complexity of the algorithm is $\Oh(m^2 p)$.
\end{proof}
\section{Appendix for \cref{sec:NPComplete}}\label{appendix:NPComplete}

\subsection{Proof of \cref{thm:memorysmall}}\label{app:memorysmall}



\memorysmall*

\begin{proof}[Proof of \cref{thm:memorysmall}]
We first define the anchoring set of players $\Lambda$ given a strategy profile~$\bsigma$.
\subparagraph*{Definition of $\Lambda$.}

\finiteMemAbstraction*

\begin{claimproof}
    We define the labelling $\Lambda$ inductively. 

    \subparagraph*{Base case.}
    First, on the one-vertex history $v_0$, we define $\Lambda(v_0) = \Pi$: at the start, all players must be anchored.
    Let us notice that the history $v_0$ satisfies Property~\ref{itm:optimistanchor}, which states that the optimists get the optimistic expectation they are supposed to get, and Property~\ref{itm:pessimistanchor}, which states that the pessimists have no profitable deviations.
    The other properties will be checked in the inductive case.

    \subparagraph*{Inductive case.}
    Suppose $\Lambda(hv)$ has already been defined, where $hv$ is a history compatible with the strategy profile $\bsigma$, and that the five properties are satisfied by $\Lambda$ on all histories on which it is already defined.
    Let $w_1, \dots, w_k$ be the successors of $v$ that are chosen by the strategy $\bsigma(hv)$ with non-zero probability, that is, the support of $\bsigma(hv)$.
    If $k=1$, then we define $\Lambda(hvw_1) = \Lambda(hv)$.
    Note that Properties~\ref{itm:splitsetsanchorwithouti},~\ref{itm:splitsetsanchorwithi}, and~\ref{itm:nosplit} are immediately satisfied, and that Properties~\ref{itm:optimistanchor} and~\ref{itm:pessimistanchor} are satisfied by induction hypothesis.
    
    If $k>1$, we need to partition the set $\Lambda(h)$ between the $k$ successors.
    To do so, we will use the following claim.
    
    \begin{claim}\label{claim:successorAnchor}
    For each player $i \in \Lambda(hv)$ that does not control the vertex $v$, the follwing holds.
    \begin{itemize}
        \item If player $i$ is an optimist, and $\X_i(\bsigma_{\|hv}) = z_i$, then there is at least one successor $w_\l$ such that $\X_i(\bsigma_{\|hvw_\l}) = z_i$.
    
        \item If player $i$ is a pessimist, then there is a successor $w_\ell \in \Supp(\bsigma(hv))$, such that for every  strategy $\tau_i^\ell$ from $w_\l$, we have $\X_i(\bsigma_{-i\|hvw_\ell}, \tau_i) \leq z_i$.
    \end{itemize}
    \end{claim}
    
    \begin{claimproof}
        The first case follows from Property~\ref{itm:optimistanchor} in the induction hypothesis.
    
        As for the second case, we proceed by contradiction.
        Let us assume that for each $w_\l$, there exists a strategy $\tau_i^\l$ such that $\X_i(\bsigma_{-i\|hvw_\l}, \tau^\l_i) > z_i$.
        Then, the strategy $\tau_i$ defined by $\tau_{i\|vw_\l} = \tau^\l$ for each $\l$ is such that $\X_i(\bsigma_{-i\|hv}, \tau_j) > z_i$, which is impossible since $\Lambda(hv)$ is assumed to satisfy Property~\ref{itm:pessimistanchor}.
    \end{claimproof}

    We define each set $\Lambda(hvw_\l)$ by iterating through each element of $\Lambda(hv)$ as follows:
    \begin{itemize}
        \item \textbf{Initialisation.} For all $w_\l$, declare $\Lambda(hvw) = \emptyset$.
        \item \textbf{Iteration over players.} Consider each player $i\in \Lambda(hv)$ sequentially and proceed as follows:
        If $i$ controls the vertex $v$, then add $i$to every set $\Lambda(hvw)$.
        If $i$ does not control $v$, then, by the claim stated earlier, there exists a successor $w_\l$. In this case, add ii to $\Lambda(hvw_\l)$ corresponding to this specific $w_\l)$
    \end{itemize}
    
    Not that the first four properties are thus guaranteed to be satisfied.

    Moreover, when there are several successors $w_\l$ possible, we always favour those such that, at the moment where the decision is taken, the sets $\Lambda(hvw_\l)$ are the smallest.
    This suffices to guarantee Property~\ref{itm:nosplit}.
    \end{claimproof}

\paragraph*{Construction of the strategy profile $\bsigma^\star$}
Based on $\Lambda$ as in \cref{lm:Lambda}, we construct a strategy profile that finite memory states. 
Formally, the definition of the finite-memory strategy $\bsigma^\star$ will be done by defining its memory structure.

The memory states are the following:
    \begin{itemize}    
        \item for each player $i$, the state $\punish_i$;

        \item for each vertex $v$ and each subset $A \subseteq \Pi$ of players such that there exists $h$ with $A = \Lambda(h)$, the state $\anchor_{Av}$;

        \item the state $\anchor_{\Pi\bot}$.
    \end{itemize}
Observe that there are at most $2^p + \Oh(p)$ such memory states.
We now define the transitions from each of those states from each vertex. Observe that long as the memory state does not change, strategy profile corresponds to a positional strategy profile. So, we describe such memoryless strategy profiles and also describe when the memory state changes. 
For each set $A \subseteq \Pi$, we write $W_A$ for the set of vertices that $\bsigma$ may visit while anchoring the set $A$, i.e., the set of vertices $v$ such that there exists a history $hv$ with $\Lambda(hv) = A$.

    \subparagraph*{Punishing memory states $\punish_i$.}
    First, let us define what $\bsigma^\star$ does when in state $\punish_i$, for some player $i$. Those memory states will correspond to the \emph{punishing strategies}, followed when player $i$ deviates from the assigned strategy with the other memory states. 
    By \cref{lm:secretlemma}, there is a positional strategy profile $\btau^{\dag i}_{-i}$ that minimises, from every vertex of the game, the payoff that player $i$ can enforce.
    In addition, we pick an arbitrary positional strategy $\sigma^{\dag i}_i$.
    Then, when the strategy profile $\bsigma^\star$ is in the memory state $\punish_i$ on reads a vertex $v$, the memory update function outputs the vertex $\btau^{\dag i}(v)$ and the same memory state $\punish_i$.

    \subparagraph*{Anchoring states with no player to anchor.}
    Let us now define what happens in memory state $\anchor_{\emptyset v}$.
    Consider the objective of achieving a payoff vector that has positive probability to be achieved in $\bsigma$, while visiting only vertices of $W_\emptyset$.
    Using classical attractor-based proofs (or \cref{lm:secretlemma}), there exists a positional strategy profile that achieves that objective with probability $1$ from every vertex from which that is possible: let us call it $\btau^{\anch \emptyset}$.
    Then, when in memory state $\anchor_{\emptyset v}$ and reading the vertex $w$, we distinguish two cases.
    \begin{itemize}
        \item If $vw$ is an edge that is compatible with the strategy profile $\btau^{\anch \emptyset}$, then the strategy profile $\bsigma^\star$ outputs the vertex $\btau^{\dag i}(w)$, where $i$ is the player controlling $v$, and shifts to the memory state $\punish_i$.

        \item Otherwise, it outputs the vertex $\btau^{\anch \emptyset}(w)$ and moves to the memory state $\anchor_{\emptyset w}$.
    \end{itemize}

    \subparagraph*{Anchoring state with one player to anchor.}
    We can now move to singletons, and define what happens in the states of the form $\anchor_{\{i\} v}$.
    In such a state, we define the strategy based on the that gives player $i$ exactly the extreme risk measure $z_i$.
    More precisely, we want player $i$ to receive payoff $z_i$ with positive probability, and never leave the set of vertices $W_{\{i\}}$ with probability~$1$.
    Using \cref{lm:secretlemma}, there exists a positional strategy profile $\btau^{\anch i}$ that satisfies that property from every vertex from which it is possible.
    Note that that objective is in particular satisfiable, and therefore satisfied by $\btau^{\anch i}$, from every vertex $v \in W_{\{i\}}$.
    Similar to the previous step, we define the strategy profile $\bsigma^\star$ in the states of the form $\anchor_{\{i\} v}$ so that it follows the strategy profile $\btau^{\anch i}$, remembers the last vertex that was visited and uses that memory to switch to the corresponding punishing state when some player $j$ deviates.

        \subparagraph*{Anchoring states with two or more players to anchor.}
    Now, let us consider the states of the form $\anchor_{A v}$, where $A$ has cardinality at least $2$.
    The existence of $\anchor_{A v}$ implies that there is a history $h$ such that $\Lambda(h) = A$.
    Moreover, since each randomisation splits the label of histories in sets that have at most one element in common (Properties~\ref{itm:splitsetsanchorwithouti} and~\ref{itm:splitsetsanchorwithi}), there is only one side of each split that can contain $A$, which implies that among such histories $h$, we can choose one that is a prefix of all others.
    After history $h$, the histories labelled by $A$ form a sequence $h, hv_1, hv_1v_2, \dots$ which may be infinite, end in a terminal vertex, or end with a new split.

    \textbf{If that sequence end with a split,} then there is a longest history $hv_1 \dots v_q$ with $\Lambda(hv_1 \dots v_q) = A$ and $k \geq 2$ vertices $w_1, \dots, w_k \in \Supp(\bsigma(hv_1 \dots v_q))$ such that we have $\Lambda(hv_1 \dots v_q w_\l) \neq \emptyset, A$ for each $\l$.
    We can then define a simple history $h'v_q$ that also goes from the vertex $\last(h)$ to the vertex $v_q$, with $\Occ(h'v_q) \subseteq \Occ(\last(h) v_1 \dots v_q)$.
    We then define the strategy profile $\bsigma^\star$ in each state $\anchor_{Av}$ so that it follows the history $h' v_q$ and remembers the last vertex visited, and switches to the state $\punish_i$ and follows the strategy profile $\btau^{\dag i}$ when a given player $i$ deviates and takes an edge that they are not supposed to take.
    Moreover, when an edge is taken that does belong to the history $h'v_q$, but not because a player deviated (it is then necessarily because of a stochastic vertex), the memory switches to the state $\anchor_{\emptyset v}$ (where $v$ is the last vertex seen) and immediately follows the corresponding strategy.
    Finally, when the vertex $v_q$ is reached and the memory is in state $\anchor_{A \last(h')}$, the strategy profile $\bsigma^\star$ chooses randomly between the edges $v_q w_1$, \dots, and $v_q w_k$, all with positive probability.
    Such action will often be referred to as a \emph{split}.

        \textbf{If that sequence is infinite or end in a terminal vertex,} then $\pi^A = v_1 v_2 \dots$ is a play, and satisfies $\Lambda(h\pi^A_{< k}) = A$ for each $k$.
        We can then consider a play $\pi^{A\star}$ with $\Occ(\pi^{A\star}) \subseteq \Occ(\pi^A)$ and $\Inf(\pi^{A\star}) \subseteq \Inf(\pi^A)$ that is either a simple path from $\pi^A_0$ to the terminal reached by $\pi^A$, or, if $\pi^A$ is infinite, a simple lasso (i.e., a play of the form $h'c^\omega$, where the history $h'c$ is simple).
        We can moreover choose $\pi^{A\star}$ so that the set of vertices visited infinitely often (if there are any) in $\pi^{A\star}$ is included in the set of vertices visited infinitely often in $\pi^A$.
Then, we can define $\bsigma^\star$ in the states of the form $\anchor_{A v}$ as following the play $\pi^{A\star}$, and remembering the last vertex seen.
    When a player $i$ deviates and takes an edge that should not be taken, the memory switches to the state $\punish_i$ and follows the strategy profile $\btau^{\dag i}$.
    Finally, when an edge is taken that does not belong to $\pi^{A\star}$ but does not correspond to a deviation either, we switch to the state $\anchor_{\emptyset w}$ where $w$ is the last vertex seen, and to the corresponding strategy profile.

\subparagraph*{Initialisation.}
    The strategy profile $\bsigma^\star$  has the state $\anchor_{\Pi\bot}$ as the initial memory state.
    In this state, it behaves exactly as in any state of the form $\anchor_{\Pi v}$, but without having memorised a last visited vertex $v$, since there is no such vertex.
    From that memory state therefore, it necessarily reads the vertex $v_0$, and starts acting as described in the previous case.

\subparagraph*{The pure case.}
In this construction, the vertices on which the strategy profile $\bsigma^\star$ proceeds to an actual randomisation (i.e., the vertices $v$ such that there exists a history $hv$ such that the support of the distribution $\bsigma^\star(hv)$ contains more than one element) are vertices on which $\bsigma$ also proceeds to such a randomisation.
Therefore, if $\bsigma$ is pure (i.e., if randomisations occur only on stochastic vertices), so is $\bsigma^\star$.

\paragraph*{A combinatorial break: counting states}

    Now that the strategy $\bsigma^\star$ is defined, let us bound the memory it uses.
    There are, obviously, exactly $p$ states of the form $\punish_i$, and one state $\anchor_{\Pi\bot}$.
    To prove that there are at most $3np-2n$ states of the form $\anchor_{Av}$, we need to prove that there are at most $3p-2$ sets $A$ such that there is a history $h$ with $\Lambda(h) = A$.

    Let us call \emph{$\Lambda$-anchored} all such sets $A$.
    By analogy with strategies, we write $\Lambda_{\|hv}$ for the labelling that maps each history $h' \in \Hist\Game_{\|v}$ compatible with $\bsigma_{\|hv}$ to the set $\Lambda(hh')$, and we will also use the notion of anchoredness for each of those labellings $\Lambda_{\|hv}$.
    We proceed by proving the following stronger result.
    
    \begin{proposition}\label{prop:combinatorial}
        For every history $h$ compatible with $\bsigma$, if $\Lambda(h)$ contains at least two elements, then there are at most $3|\Lambda(h)|-2$ sets that are $\Lambda_{\|h}$-anchored.
    \end{proposition}

\begin{claimproof}
    For each history $h$, we write $f(h)$ for the number of $\Lambda_{\|h}$-anchored sets that have cardinal at least $2$.
    There are $|\Lambda(h)|+1$ subsets of $\Lambda(h)$ that have cardinality $0$ or $1$: the result will therefore be proved if we prove $f(h) \leq 2|\Lambda(h)| - 3$.
    The proof goes by induction on $m = \Lambda(h) \geq 2$.

    \subparagraph*{Base case.}
    If $m = 2$, the set $\Lambda(h)$ is a pair $\Lambda(h) = \{i, j\}$.
    Then, since the $\Lambda_{\|h}$-anchored sets are all subsets of $\Lambda(h)$, the only set of cardinality at most $2$ that is $\Lambda_{\|h}$-anchored is the pair $\{i, j\}$ itself, hence we have $f(h) = 2 \times 2 - 3 = 1$, as desired.

    \subparagraph*{Inductive case.}
    If $m > 2$, and if we assume that the result is true for every history $h'$ with $2 \leq |\Lambda(h')| \leq m-1$, then let $\{v_1, \dots, v_k\} \subseteq \Supp(\bsigma(h))$ be the set of possible next vertices $v$ such that $|\Lambda(hv)| \geq 2$.

    If $k = 1$, i.e., if $\Lambda(hv_1) = \Lambda(h)$, then we have $f(h) = f(hv_\l)$ and the result for $h$ will be proved if we prove it for $hv_1$.
    Following that reasoning, we can extend the history $h$ until we are not in that case: if we always are, then the only $\Lambda_{\|h}$-anchored sets are $\Lambda(h)$ itself, and possibly the empty sets and some singletons, hence $f(h) = 1$ and the result is immediate.
    We can therefore assume that $k > 1$.

    Let $i$ be the player controlling the vertex $\last(h)$; we set $i = \bot$ if $\last(h)$ is a stochastic node.
    Then, by Properties~\ref{itm:splitsetsanchorwithouti} and~\ref{itm:splitsetsanchorwithi} of \cref{lm:Lambda}  guaranteed during the construction of $\Lambda$, the sets $\Lambda(hv_1) \setminus \{i\}, \dots, \Lambda(hv_k) \setminus \{i\}$ form a partition of $\Lambda(h) \setminus \{i\}$.
    Therefore, no set of cardinality at least $2$ can be simultaneously $\Lambda(hv_\l)$-anchored and $\Lambda(hv_{\l'})$-anchored for $\l \neq \l'$, hence the equality $f(h) = 1+\sum_\l f(hv_\l)$.
    Now, since each set $\Lambda(hv_\l)$ has at least $2$ and less than $m$ elements, we can apply the induction hypothesis to deduce:
    $$f(h) \leq 1 + \sum_{\l=1}^k (2|\Lambda(hv_\l)| - 3).$$
    Moreover, we have $\sum_\l |\Lambda(hv_\l)| \leq m + k - 1$ (each element of $\Lambda(h)$ occurs in one of the sets $\Lambda(hv_\l)$, except possibly one that would occur in all of them), hence the inequality above becomes:
    $$f(h) \leq 1 + 2(m+k-1) -3k$$
    $$= 2m - 1 - k$$
    and since we have assumed $k \geq 2$, we obtain $f(h) \leq 2m-3$.
\end{claimproof}    

Let us recall that $\Lambda(v_0) = \Pi$.
As a particular case of this claim, we obtain, if $p \geq 2$, that there are at most $3p-2$ sets that are $\Lambda$-anchored, as desired.
In the case $p = 1$, the game $\Game_{\|v_0}$ is an MDP, and using \cref{lm:secretlemma}, we can immediately construct $\bsigma^\star$ as a positional strategy (which has therefore $1 \leq 3n \times 1 - 2n + 1 + 1$ memory states) with the same risk measure as $\bsigma$.

    \paragraph*{The strategy profile $\bsigma^\star$ has the desired extreme risk measures.}

We now show that $\X(\bsigma^\star) = \X(\bsigma)$.
Let us recall that we defined $\bz = \X(\bsigma)$.
    
\begin{proposition}\label{prop:ActualPayoff}
    The strategy profile $\bsigma^\star$ satisfies the equality $\X(\bsigma^\star) = \bz$.
\end{proposition}

\begin{claimproof}
    Let $i$ be a player: we want to prove that $\X(\bsigma^\star) = z_i$.
    Let us first see how $z_i$ has positive probability of being obtained in the strategy profile $\bsigma$, and we will then show that no larger (respectively smaller) if $i$ is optimistic (respectively pessimistic) has a positive probability by the same strategy.

    \subparagraph*{Player $i$ gets payoff $z_i$ with positive probability.}
    Let $A \subseteq \Pi$ be one of the smallest sets (for the inclusion relation) containing $i$ such that there exists a history $hu$ with $\Lambda(hu) = A$.
    Then, by construction of $\Lambda$, there exists a finite sequence of sets $\Pi = A_0, A_1, \dots, A_m = A$ and of histories $h_1 v_1 w_1, \dots, h_m v_m w_m$ where for each $k$, the history $h_{k+1}$ starts from $w_k$, the history $h_1 v_1 \dots h_k v_k w_k$ is compatible with $\bsigma$, and we have $\Lambda(h_1 v_1 \dots h_k v_k) = A_{k-1}$ and $\Lambda(h_1 v_1 \dots h_k v_k w_k) = A_k$.
    We can then write $hu = h_1 v_1 \dots h_m v_m w_m$.
    
    Consider the strategy profile $\bsigma^\star$, which initially follows  the positional strategy profile $\btau^{\anch \Pi}$. This strategy profile generates, with nonzero probability, a history $h'_1 v_1$ starting from vertex $v_0$ to vertex $v_1$, based on our construction. 
    From that vertex $v_1$, it proceeds to a randomised action and, with positive probability, moves to the vertex $w_1$ and switches to the positional strategy profile $\btau^{\anch A_1}$, and so on: there is, therefore, a history $h'_1 v_1 h'_2 v_2 \dots h'_m v_m w_m$ that is compatible with the strategy profile $\bsigma^\star$ and after which the collective memory is in the state $\anchor_{A v_m}$, and plays accordingly.

    Since $A_m= A$ is the  subset of $\Pi$ where $i\in A = \Lambda(hu)$, we are in the case where the set $A$ is no longer split further by our labelling. That is, there is a play $\pi$ from $w_m$ such that $\Lambda(h \pi_{\leq k}) = A$ for every $k$ such that that is defined.
    Then, in the construction of the strategy profile $\bsigma^\star$ we have distinguished two cases: the one where $A$ was a singleton, and the one where it had at least two elements (the empty case is excluded, since $A$ contains the player $i$).

    \textbf{If $A$ is a singleton,} then after the history $hu$, without any player deviating, all players are following the strategy profile $\btau^{\anch i}$.
    By its definition, that strategy profile achieves the payoff $z_i$ for player $i$ with positive probability.

    \textbf{If $A$ has at least two elements,} then after that same history, all players are following the play $\pi^{A\star}$, which yields the same payoffs as $\pi^A$.
    However, we must still prove that player $i$ actually gets the payoff $z_i$ in $\pi^A$, and that the play $\pi^{A\star}$ is generated with positive probability (i.e. that it does not cross infinitely many stochastic vertices---which we must first show for $\pi^A$).
We do so in the following claim, which we will use again later.

\begin{claim}\label{claim:piA}
    The play $\pi^A$ is (eventually) generated with positive probability when the players follow the strategy profile $\bsigma$.
    Similarly, the play $\pi^{A\star}$ is generated with positive probability when they follow $\bsigma^\star$.
    Both plays yield to each player $j \in A$ the payoff $z_j$.
\end{claim}

\begin{claimproof}
    Let $j \in A$.
    Let us proceed by case disjunction according to the risk measure used by player $j$.

\emph{If player $j$ is an optimist,} then, by Property~\ref{itm:optimistanchor} of \cref{lm:Lambda}, we have $\X_j(\bsigma_{\|h}) = z_j$, and therefore $\prob_{\bsigma_{\|h}}(\mu_j = z_j) > 0$, i.e., by the law of total probability:
        $$\prob_{\bsigma_{\|hu}}(\pi^A) \prob_{\bsigma_{\|h}}(\mu_j = z_j \mid \pi^A) + \sum_k \sum_{w \in E(\pi_k) \setminus \{\pi_{k+1}\}} \prob_{\bsigma_{\|hu}}(\pi^A_{\leq k} w) \prob_{\bsigma_{\|hu}}(\mu_j = z_j \mid \pi^A_{\leq k}w) > 0.$$
        But using Property~\ref{itm:nosplit}, all the terms of the summation on the right are zero, hence the product $\prob_{\bsigma_{\|h}}(\pi^A) \prob_{\bsigma_{\|h}}(\mu_j = z_j \mid \pi^A) > 0$ is positive, i.e. the play $\pi^A$ has  a positive probability of being generated and $\mu_j(\pi^A) = z_j$.

   \emph{If player $j$ is a pessimist,} then because of Property~\ref{itm:nosplit} again, for every $k \geq 0$ and each $w \in \Supp\left(\bsigma\left(h\pi^A_{\leq k}\right)\right)$, there exists a strategy $\tau_j^{kw}$ such that $\X_j(\bsigma_{-j\|h\pi^A_{\leq k}w}, \tau_j^{kw}) > z_j$.
        By composing all those strategies, we obtain a deviation $\tau_j$ of the strategy $\sigma_{j\|h}$; which, by Property~\ref{itm:pessimistanchor}, satisfies the inequality $\X_j(\bsigma_{-j\|h}, \tau_j) \leq z_j$.
        Therefore, either:
        \begin{itemize}
            \item we have:
            $$\min_k \min_{w \in \Supp\left(\bsigma\left(h\pi^A_{\leq k}\right)\right) \setminus \{\pi^A_{k+1}\}} \X_j(\bsigma_{-j\|h\pi_{\leq k}w}, \tau^{kw}_j) \leq z_j,$$
            which is impossible by definition of the strategies $\tau_j^{kw}$;

            \item or we have $\prob_{\bsigma_{-j\|h}, \tau_j}(\pi^A) = \prob_{\bsigma_{\|h}}(\pi^A) \neq 0$ and $\mu_j(\pi^A) \leq z_j$, and then actually $\mu_j(\pi^A) = z_j$.
        \end{itemize}

We have thus proven that player $j$ gets the payoff $z_j$ in $\pi^A$, and that the play $\pi^A$ is generated with positive probability in $\bsigma$.
The analogous results about $\pi^{A\star}$ follow using the equalities $\Occ(\pi^{A\star}) = \Occ(\pi^A)$ and $\Inf(\pi^{A\star}) = \Inf(\pi^A)$.
\end{claimproof}

    In those two cases (if $A$ is a singleton or has several elements), we obtain that the strategy profile $\bsigma^\star$ is such that, with some positive probability, player $i$ gets the payoff $z_i$.

    \subparagraph*{Player $i$ gets risk measure $z_i$.}
    We still have to prove that player $i$ has zero probability of getting a lower payoff (if they are a pessimist) or a higher payoff (if they are an optimist).
    To show both cases, we prove the following claim:

    \begin{claim}
        Every payoff vector that has a positive probability of being achieved in the strategy profile $\bsigma^\star$ also has a positive probability of being achieved in the strategy profile $\bsigma$.
    \end{claim}

\begin{claimproof}
    Let $\bz'$ be such a payoff vector.
    Then, there is a history $hw$ compatible with $\bsigma^\star$ and a set $A \subseteq \Pi$ such that, after the history $hw$, the strategy profile $\bsigma^\star$ is in state $\anchor_{A \last(h)}$, and from that point it has a nonzero probability of achieving the payoff vector $\bz'$ while staying in states of the form $\anchor_{A v}$.
    
    \emph{If $A$ is empty}, then the strategy profile $\btau^{\anch \emptyset}$ has been defined as a strategy profile that almost surely generates a payoff vector that is generated with positive probability by $\bsigma$, from every vertex from which that is possible.
    That requirement is satisfiable, and therefore satisfied by $\btau^{\anch \emptyset}$, from the vertex $w$, since that vertex is itself reached with positive probability in the strategy profile $\bsigma$.
    Therefore, the payoff vector $\bz'$ is also achieved with positive probability in $\bsigma$.
    
    \emph{If $A$ is a singleton,} say $A = \{j\}$, then the strategy profile $\btau^{\anch j}$ has been defined so that from every vertex from which that is possible, on the one hand, it generates the payoff $z_j$ with positive probability, and on the other hand, it is almost sure that the payoff vector that will be generated has also positive probability to be generated in $\bsigma$.
    Similarly as above, that requirement is satisfiable from the vertex $w$, since $\bsigma_{\|hw}$ satisfies it.
    Therefore, again, the payoff vector $\bz'$ is also achieved with positive probability in $\bsigma$.
    
    \emph{If $A$ has at least two elements}, then the strategy profile $\bsigma^\star$ stays in states of the form $\anchor_{A v}$ only along one play, namely $\pi^{A\star}$, and that play generates a payoff vector that was also associated with the play $\pi^A$ that by \cref{claim:piA}, has positive probability to be generated in $\bsigma$, hence the same conclusion.
\end{claimproof}

This proves the equality $\X_i(\bsigma^\star) = z_i$.
\end{claimproof}

\paragraph*{The strategy profile $\bsigma^\star$ is an XRSE.}

We have now constructed the finite-memory strategy profile $\bsigma^\star$, showed that it had the expected number of memory states, and that it generates the expected risk measures.
We must now give the final argument for our construction: that strategy profile is also an extreme risk-sensitive equilibrium.
We will prove that result by showing separately that optimists have no profitable deviations, and then that neither do pessimists.

\begin{proposition}\label{prop:NodeviationOpt}
    No optimist has a profitable deviation in $\bsigma^\star$.
\end{proposition}

\begin{claimproof}
    Let $i$ be an optimist, and let us consider a deviation $\sigma'_i$ of that player from $\bsigma^\star$. Let us write $z'$ for the risk measure $z' = \X_i(\bsigma^\star_{-i}, \sigma'_i)$.

    Let us notice that along every play compatible with $\bsigma^\star_{-i}$, the transitions that are possible in the memory structure of the strategy profile $\bsigma^\star$ can be classified as follows:
    \begin{itemize}
        \item transitions among states of the form $\anchor_{A v}$ for a fixed $A$;
        
        \item transitions from a state of the form $\anchor_{A v}$ to a state of the form $\anchor_{B w}$ with $B \subset A$;

        \item transitions from a state of the form $\anchor_{A v}$ to the state $\punish_i$;

        \item and transitions from $\punish_i$ to itself
    \end{itemize}
    Therefore, any such play stabilises either in the state $\punish_i$, or among the states of the form $\anchor_{A v}$ for a fixed set $A$.
    Consequently, if in the strategy profile $(\bsigma^\star_{-i}, \sigma'_i)$ player $i$ gets the payoff $z'$ with positive probability, then we can also say that either:
    \begin{itemize}
        \item with positive probability, player $i$ gets the payoff $z'$ \emph{and} the state $\punish_i$ is reached;

        \item or there exists a set $A \subseteq \Pi$ such that with positive probability, player $i$ gets the payoff $z'$, and the collective memory remains in states of the form $\anchor_{A v}$.
    \end{itemize}

    \emph{In the first case,} let us consider a history $hv$ compatible with $\bsigma^\star_{-i}$ such that the collective memory is in an anchoring state after $h$ and in state $\punish_i$ after $hv$.
    If player $i$ can obtain the risk measure $z'$ by going to $v$ from that vertex against $\bsigma^\star_{-i\|h}$, and therefore, against the punishing strategy profile $\btau^{\dag i}_{-i}$, it means that they can enforce that risk measure against every possible strategy profile from $\last(h)$.
    On the other hand, if the collective memory is in an anchoring state after $h$, it means that the vertex $\last(h)$ is also visited with positive probability in the strategy profile $\bsigma$ (otherwise we would have switched to a punishing state earlier).
    There is therefore a history $h'$ compatible with $\bsigma$ such that $\last(h) = \last(h')$; and after that history, against the strategy profile $\bsigma_{\|h'}$, player $i$ also has the possibility of getting with positive probability the payoff $z'$.
    Since $\bsigma$ is an XRSE, that implies $z' \leq z_i$.

    \emph{In the second case,} let us notice that the strategy profiles of the form $\btau^{\anch A}$ are pure, and therefore that any deviation of player $i$ is immediately detected and leads to a switch to state $\punish_i$.
    Therefore, if the collective memory remains in states of the form $\anchor_{A v}$, it means that player $i$ is actually following the strategy $\sigma^\star_i$.
    Thus, we also have $z' \leq z_i$.
    
    The strategy $\sigma'_i$ is not a profitable deviation from $\bsigma^\star$.
\end{claimproof}

We can now end the proof with the dual proposition.

\begin{proposition}
    No pessimist has a profitable deviation in $\bsigma^\star$.
\end{proposition}

\begin{claimproof}
    Let $i$ be a pessimist, and consider a deviation $\sigma'_i$ of that player from $\bsigma^\star$.
    We intend to prove that the deviation $\sigma'_i$ is not profitable, that is, when following the strategy profile $(\bsigma^\star_{-i}, \sigma'_i)$, there is still a positive probability that player $i$ receives a payoff smaller than or equal to $z_i$.
    Using \cref{lm:secretlemma}, we can assume without loss of generality that $\sigma'_i$ is pure.

    First, we observe that for each history $hv$ compatible with $\bsigma^\star$ such that, after $hv$, the collective memory is in state $\anchor_{A \last(h)}$ with $i \in A$, the vertex $v$ is such that there also exists a history $h'v$ compatible with $\bsigma$ with $\Lambda(h'v) = A$.
    By Property~\ref{itm:pessimistanchor} of \cref{lm:Lambda}, we have $\X_i(\bsigma_{-i\|h'v}, \tau_i) \leq z_i$ for every $\tau_i$.
    Therefore, if player $i$ accepts to follow the history $hv$ and, then, deviates and takes an edge that makes the collective memory switch to the state $\punish_i$, then with positive probability player $i$ gets a payoff lesser than or equal to $z_i$.
    If such an action is ever performed, then the deviation $\sigma'_i$ is not profitable.

    Let us now assume that $\sigma'_i$ performs no such action: after every history $hv \in \Hist_i \Game_{\|v_0}$, if the collective memory is in a state of the form $\anchor_{A \last(h)}$ with $i \in A$, the vertex $\sigma'_i(hv)$ belongs to the set $\Supp(\sigma^\star_i(hv))$.
    Then, by Property~\ref{itm:splitsetsanchorwithi} of \cref{lm:Lambda}, we also have $i \in \Lambda(hv\sigma'_i(hv))$.
    Thus, there still exists a set $A$ with $i \in A$ such that, with positive probability, when following the strategy profile $(\bsigma^\star_{-i}, \sigma'_i)$, the strategy profile $\bsigma^\star_{-i}$ stabilises among memory states of the form $\anchor_{A v}$; and then, the strategy profile $\bsigma^\star$ only proceeds to pure actions, hence the strategy $\sigma'_i$ is actually following $\sigma^\star_i$.

    Using the same arguments as in the proof of \cref{prop:ActualPayoff} (definition of $\btau^{\anch i}$ in case $A = \{i\}$, and \cref{claim:piA} in case $A$ has more elements), we can then conclude that the player $i$ gets the payoff $z_i$ with positive probability and therefore the deviation $\sigma'_i$ is not profitable.
    \end{claimproof}

    The strategy profile $\bsigma^\star$ is an XRSE, satisfies the equality $\X(\bsigma^\star) = \X(\bsigma)$, and uses the desired number of memory states.
    Furthermore, if $\bsigma$ is pure, so is $\bsigma^\star$.
\end{proof}
\subsection{Proof of \cref{lemma:np_hardness}}\label{app:np_hardness}
\NPHard*
\begin{proof}[Proof of \cref{lemma:np_hardness}]
  We prove $\NP$-hardness by reducing from the problem $\THREESAT$. Consider a $\THREESAT$ formula $\Phi$, over the variables  $x_1,\dots,x_n$, where $\Phi = C_1\land C_2\land \dots\land C_m$, where for each $i$ we have $C_i = (\ell_{i1}\lor \ell_{i2}\lor \ell_{i3})$ and for $j=1,2,3$, we have $\ell_{ij} = x_k$ or $\ell_{ij} = \neg x_k$ for some $k\in \{1, \dots, n\}$. 
  We construct  a game $\Game_\Phi$ with two players for each literal $\ell$, denoted by $\Circle \ell$ and $\Square \ell$. The game is depicted in \cref{fig:NPhard}.
  For convenience, some terminal vertices have been represented several times.
  Each player $\Circle \ell$ controls one vertex, the vertex $\Circle \l$, of circled shape, and symmetrically, each player $\Square \ell$ controls the square-shaped vertex $\Square \l$.
  Further, we add a player $C_i$, who controls the vertex $C_i$, for each clause $C_i$. Finally, there is a player $\Diamond$ who does not control any vertex. There are also stochastic vertices, that are represented by the black circles.
  In each terminal vertex, the symbol $\forall$ should be understood as "every (other) player".


 We assume all players are pessimistic, and ask if there is an XRSE where player $\Diamond$'s risk measure is exactly $2$.
We give the formal definition of the game $\Game_\Phi$ below.

    \begin{figure}
        \centering
        
        \begin{tikzpicture}[shorten >=1pt, node distance=1.5cm and 2cm, on grid, auto, scale=1.1]
          \tikzstyle{state}=[circle, draw, minimum size=20pt, inner sep=1pt]
          \tikzstyle{squarestate}=[rectangle, draw, minimum size=20pt, inner sep=1pt] 

            \node[state] (qnc1) at (0, 0) {$\neg x_{1}$};
            \node[state, initial,initial text=] (qc1) at (0, 1) {$x_{1}$};

            \node[scale=0.6] (tpunish1) at (0,-1) {$t_\dag:~\stack{\forall}{0}$};
            
            \node[stoch, scale=0.6] (stoc2) at (1, 0) {$s_{\neg x_1}$};
            \node[stoch, scale=0.6] (stoc1) at (1, 1) {$s_{x_1}$};

            \node[scale=0.6] (reward1) at (1,2) {$f_{x_1}:~\stack{\circ x_1}{1}$$\stack{\forall}{2}$};
            \node[scale=0.6] (reward2) at (1,-1) {$f_{\neg x_1}:~\stack{\circ \neg x_1}{1}$$\stack{\forall}{2}$};
            
            \node[squarestate] (qns1) at (2, 1) {$\neg x_{1}$};
            \node[squarestate] (qs1) at (2, 0) {$x_{1}$};

            \node[scale=0.6] (punishW1) at (2,2) {$t_\diamond:~\stack{\diamond}{0}$$\stack{\forall}{2}$};
            \node[scale=0.6] (punishW2) at (2,-1) {$t_\diamond:~\stack{\diamond}{0}$$\stack{\forall}{2}$};
            
            \node[state] (qnc2) at (3.5, 0) {$\neg x_{2}$};
            \node[state] (qc2) at (3.5, 1) {$x_{2}$};
            \node[stoch, scale=0.6] (stoc3) at (4.5, 1) {$s_{x_2}$};
            \node[stoch, scale=0.6] (stoc4) at (4.5, 0) {$s_{\neg x_2}$};

            \node[scale=0.6] (reward3) at (4.5,2) {$f_{x_2}:~\stack{\circ x_2}{1}$$\stack{\forall}{2}$};
            \node[scale=0.6] (reward4) at (4.5,-1) {$f_{\neg x_2}:~\stack{\circ \neg x_2}{1}$$\stack{\forall}{2}$};
            
            \node[squarestate] (qns2) at (5.5, 1) {$\neg x_{2}$};
            \node[squarestate] (qs2) at (5.5, 0) {$x_{2}$};

            \node[scale=0.6] (punishW3) at (5.5,2) {$t_\diamond:~\stack{\diamond}{0}$$\stack{\forall}{2}$};
            \node[scale=0.6] (punishW4) at (5.5,-1) {$t_\diamond:~\stack{\diamond}{0}$$\stack{\forall}{2}$};

            \node[scale=0.6] (tpunish2) at (3.5,-1) {$t_\dag:~\stack{\forall}{0}$};
            \node (qc3) at (6.5, 1) {};

            \node (dots) at (6.8,0.5) {$\dots$};

            \node[squarestate, initial,initial text=] (qnsn) at (8, 1) {$\neg x_{n}$};
            \node[squarestate, initial,initial text=] (qsn) at (8, 0) {$x_{n}$};

            \node[scale=0.6] (punishW5) at (8,2) {$t_\diamond:~\stack{\diamond}{0}$$\stack{\forall}{2}$};
            \node[scale=0.6] (punishW6) at (8,-1) {$t_\diamond:~\stack{\diamond}{0}$$\stack{\forall}{2}$};
            
            \node[stoch, scale=0.6] (stochfin) at (9,0.5) {$s_\mathsf{r}$};
            \node (fakenode) at (9,0.5) {};
            
            \node (c1) at (10,1.8) {$C_1$};
            \node (c2) at (10,1) {$C_2$};
            \node (cdots) at (10,0.4) {$\vdots$};
            \node (c3) at (10,-0.6) {$C_m$};

            \node[scale=0.6] (ter1) at (11.2, 1.7) {$t_{x_2}:~\stack{\Box x_2}{1}$ $\stack{\forall}{2}$};
            \node[scale=0.6] (ter2) at (11.2, 1) {$t_{x_4}:~\stack{\Box x_4}{1}$ $\stack{\forall}{2}$};
            \node[scale=0.6] (ter3) at (11.2, -0.3) {$t_{\neg x_{11}}:~\stack{\Box\neg x_{11}}{1}$ $\stack{\forall}{2}$};

          \path[->]
              (qc1) edge (stoc1)
              (qc1) edge (qnc1)
              (qnc1) edge (stoc2)
              (stoc1) edge (qns1)
              (stoc2) edge (qs1)
              (qns1) edge (qc2)
              (qs1) edge (qc2)
              (qc2) edge (stoc3)
              (qc2) edge (qnc2)
              (qnc2) edge (stoc4)
              (stoc3) edge (qns2)
              (stoc4) edge (qs2)
              (qns2) edge (qc3)
              (qs2) edge (qc3)
              (qnsn) edge (stochfin)
              (qsn) edge (stochfin)
              (fakenode) edge (c1)
              (fakenode) edge (c2)
              (fakenode) edge (c3)
              (c2) edge (ter1)
              (c2) edge (ter2)
              (c2) edge (ter3)
              (qnc1) edge (tpunish1)
              (qnc2) edge (tpunish2);

        \path[->]
            (stoc1) edge (reward1)
            (stoc2) edge (reward2)
            (stoc3) edge (reward3)
            (stoc4) edge (reward4)
            (qs1) edge (punishW2)
            (qns1) edge (punishW1)
            (qs2) edge (punishW4)
            (qns2) edge (punishW3)
            (qsn) edge (punishW6)
            (qnsn) edge (punishW5);
        
        \end{tikzpicture}
        \caption{Construction of a game $\Game_\Phi$ from a $\THREESAT$ formula $\Phi$}
        \label{fig:NPhard}
    \end{figure}
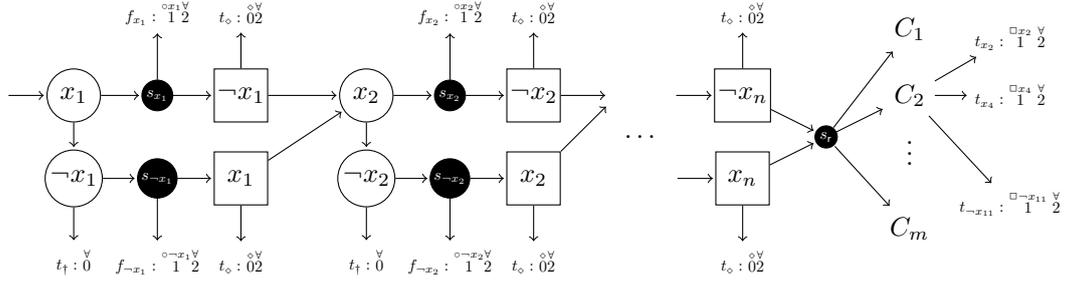


\subparagraph*{Construction of the game $\Game_\Phi$: vertices, edges and payoffs.}
For each literal $\ell$, we define two players $\Square\ell$ and $\Circle\ell$. We add one other player $C_i$ for each clause $C_i$, and an additional constraining player $\Diamond$.
All players are pessimists.

Each player owns at most one vertex in the game, and therefore, we will refer to the player and vertex interchangeably. There is one vertex for each of the players mentioned above other than $\Diamond$, who owns no vertices. Further, there are $2n + 1$ many stochastic vertices: one for each literal $s_{x_1},s_{x_2},\dots,s_{x_n}$,  $s_{\neg x_1},s_{\neg x_2},\dots,s_{\neg x_n}$, and finally one clause-randomiser $s_\mathsf{r}$. 
There are also $2n + 2$ terminal vertices, written $f_{\ell}$ and $t_{\ell}$ for each literal $\ell$, and further the terminal vertices $t_\Diamond$ and $t_\dag$.

We now define the edges between the vertices of the graph for all $i\in \{1, \dots, n\}$:  there are edges from $\Circle x_i$ to $\Circle\neg x_i$, and edges from $\Circle\neg x_i$ to $t_\dag$.
    Further, for every literal $\ell = x_i$ or $\neg x_i$, there are edges:
    \begin{itemize}
        \item from $\Circle\ell$ to $s_{\ell}$;
        \item from $s_{\ell}$ to $f_\ell$ and to $\Square\Bar{\ell}$, where $\Bar{\ell} = \neg x_i$ if $\ell = x_i$ and $\Bar{\ell} = x_i$ if $\ell = \neg x_i$;
        \item from $\Square\ell$ to $t_\Diamond$;
        \item from $\Square\ell$ to $\Circle x_{i+1}$ if  $i<n$,  and  to $s_\mathsf{r}$ if  $i=n$.
    \end{itemize}
    Finally, for all clauses $C_j$, there are edges from $s_\mathsf{r}$ to $C_j$ and from $C_j$ to $t_{\ell}$ such that $\ell$ occurs positively in the clause $C_j$.

    The terminal vertices yield the following payoffs.
\begin{itemize}
    \item In terminal $t_\ell$, all players get payoff $2$, except the player $\Square \ell$ who gets payoff $1$. 
    \item In terminal $f_\ell$, all players get payoff $2$, except player $\Circle \ell$ who gets payoff $1$.
    \item In terminal $t_\dag$, all players get payoff $0$.
    \item In terminal $t_\Diamond$, all players get payoff $2$, except player $\Diamond$ who gets payoff $0$.
\end{itemize}    

Finally, we let the constraints be that player $\Diamond$ gets a risk measure of exactly $2$.
Equivalently, we define $\bx$ and $\by$ by $\by = (2)_{i \in \Pi}$, $x_i = 0$ for each $i \in \Pi \setminus \{\Diamond\}$, and $x_\Diamond = 2$.





    \subparagraph*{If $\Phi$ is satisfiable, then there is an XRSE satisfying the constraints.}
    Consider a satisfying assignment of the $\THREESAT$ formula, described by the assignment $\alpha$ from the set of all variables to $\{\top,\bot\}$.
    
    For each $i$, let $\ell_i$ denote the literal, among $x_i$ and $\neg x_i$, which is set to true by 
    the satisfying assignment $\alpha$.
    Let us define the (positional) strategy profile $\bsigma^\alpha$.
    
    \begin{itemize}
        \item Player $\Circle \ell_i$ goes to $s_{\ell_i}$.
        \item Player $\Circle x_i$ goes to $s_{x_i}$ if $\alpha(x_i) = \top$, and to $\Circle \neg x_i$ otherwise.
        \item Player $\Circle\neg x_i$ goes to $s_{\neg x_i}$ if $\alpha(x_i) = \top$ and to $t_\dag$ otherwise.
        \item For each player $\Square \ell$, the strategy is to chose the edge that does \emph{not} lead to $t_\Diamond$. That is, the edge to $\Circle x_{i+1}$ if $\ell = x_i$ or $\neg x_{i}$ and  $i<n$,  and  the edge to $s_\mathsf{r}$ if  $i=n$.
        \item Each clause player $C_i$ takes the edge to the vertex $\ell_j$ such that the litteral $\ell_j$ was set to true by the satisfying assignment $\alpha$. 
    \end{itemize}
     We now show that this is an XRSE that satisfies the constraint. First, we verify if the constraints are satisfied. Observe that following the strategy profile $\bsigma^\alpha$, it is almost sure that none of the terminals where player $\Diamond$ has payoff less than $2$ will be reached. Therefore this satisfies the constraints. 

     We now argue that $\bsigma^\alpha$ is an XRSE, i.e. that no player can get a better risk measure by deviating.
     The result is immediate for player $\Diamond$ and for the clause players, who all get risk measure $2$, the best they could hope for.
     
     For each literal $\l$, player $\Circle\ell$ gets risk measure $1$ if $\l$ is set to true, and risk measure $2$ if $\ell$ is set to false.
     The same argument as above holds therefore in the second case.
    In the first case, they get risk measure $0$, but they have no profitable deviation, since the only deviation available leads to $t_\dag$ and to the payoff $0$.
     
     Player $\Square \ell$  has also risk measure  $2$ when $\ell$ is set to false.
     Otherwise, they get payoff $1$. In that second case, the vertex owned by the player is not visited in any history of the game, hence they have no possibility of deviating.

     The (positional) strategy profile $\bsigma^\alpha$ is therefore an XRSE.
     
    \subparagraph*{If there is an XRSE satisfying the constraints, then $\Phi$ is satisfiable.} 
    Let us assume that there exists an XRSE $\bsigma$ in the game $\Game_\Phi$, such that player $\Diamond$ gets the risk measure $2$.
    We prove, first, that we can assume that $\bsigma$ is pure (and therefore positional, since there is then only one history leading to each vertex).

\begin{claim}
    There exists an XRSE $\bsigma^\star$ in $\Game_{\|v_0}$ where player $\Diamond$ gets risk measure $2$ that is positional.
\end{claim}

\begin{proof}
    Let us first focus on what happens in vertices that have positive probability of being reached.

    If the vertex $\Circle\neg x_i$ has a positive probability of being reached in $\bsigma$, then any strategy of the player $\Circle \neg x_i$ that goes to $t_\dag$ with positive probability gives the player $\Diamond$ the risk measure $0$.
    Therefore, necessarily, the strategy $\sigma_{\circ \neg x_i}$ consists of deterministically going to $s_{\neg x_i}$.
    The same argument holds for the vertices of the form $\Square \l$.
    
    If now the vertex $\Circle x_i$ has a positive probability of being reached and if the player $\Circle x_i$ randomises between the two edges available, then she gets the risk measure $1$, since the terminal vertex $f_{x_i}$ is reached with positive probability and $t_\dag$ with probability zero.
    But then, if she deviates and goes to the vertex $\Circle \neg x_i$ with probability $1$, she avoids the terminal vertex $f_{x_i}$, and the other players will not react since they do not detect the deviation.
    She therefore gets the risk measure $2$, and the deviation is profitable.
    Consequently, the strategy $\sigma_{\circ x_i}$ can only  deterministically select one of those two edges.

    At the end of the game, for each $j$, the player $C_j$ could play a randomised strategy. In such a case, her strategy can be replaced by a pure strategy that takes, deterministically, one of the edges that she was previously taking.
    Such a modification in her strategy can only increase the risk measure of some players (namely, those of the form $\Square \l$) without impacting player $\Diamond$'s risk measure or giving any player the possibility of profitably deviating.

    Finally, if one of those vertices is reached after a history that is not compatible with $\bsigma$, i.e. if one of those players deviates: it is necessarily due to a deviation of a player of the form $\Circle x_i$, since any other deviation would immediately lead to a terminal vertex.
    If she went to $s_{x_i}$ instead of $\Circle \neg x_i$, what the other players do afterwards does not matter, since such a deviation cannot be profitable: with positive probability, the terminal vertex $f_{x_i}$ is reached, and she gets payoff $1$.
    If she went to $\Circle x_i$ instead of $s_{x_i}$, then we can assume that player $\Circle \neg x_i$'s strategy consists of going to the terminal vertex $t_\dag$, giving her the payoff $0$.
    Those modifications do not impact the fact that $\bsigma$ is an XRSE.
\end{proof}

We therefore assume that $\bsigma$ is positional.
Let us now define the assignment $\alpha$ as follows: for each variable $x_i$ we have $\alpha(x_i) = \top$ if $\sigma_{\circ x_i}(\Circle x_i) = s_{x_i}$, and $\alpha(x_i) = \bot$ if $\sigma_{\circ x_i}(\Circle x_i) = \Circle \neg x_i$.
Let then $C_j$ be a clause, and let us prove that it is satisfied by $\alpha$.
Let $t_\l = \sigma_{C_j}(C_j)$.
Then, the player $\Square \l$ gets risk measure $1$ in the XRSE $\bsigma$.
Consequently, the vertex $\Square \l$ is never reached: otherwise, the only play compatible with $\bsigma$ in which player $\Square \l$ gets payoff $1$ would traverse the vertex $\Square \l$, and player $\Square \l$ would have a profitable deviation by going to the terminal vertex $t_\Diamond$.
If that is the case, then the definition of $\alpha$ given above implies that the literal $\l$ is true.

The assignment $\alpha$ satisfies therefore the formula $\Phi$.

\subparagraph*{Conclusion.}
We have defined an instance of the constrained existence problem of XRSEs from an instance of $\THREESAT$ and proved that one is a positive instance if and only if the other is.
This proves the $\NP$-hardness of the constrained existence problem of XRSEs, since the game $\Game_\Phi$ can clearly be constructed in polynomial time.
Moreover, the game $\Game_\Phi$ is such that if an XRSE where player $\Diamond$ gets risk measure $2$ exists, then there also exists such an equilibrium that it positional, which proves also $\NP$-hardness when the players are restricted to pure, stationary or positional strategies.
\end{proof}


\subsection{Proof of \cref{lm:ptimeupperbound}}\label{app:ptimeupperbound}

\ptimeupperbound*

\begin{proof}[Proof of \cref{lm:ptimeupperbound}] \paragraph*{Preliminary remarks}

We are given the game $\Game_{\|v_0}$ and two threshold vectors $\bx, \by \in \Qb^{\Pi}$; we wish to find an XRSE $\bsigma$ such that $\bx \leq \X(\bsigma) \leq \by$. 

    Throughout the proof, when $W \subseteq V$ is a set of vertices and $F \subseteq E$ is a set of edges, we write $\Attr(W, F)$ for the \emph{positive probabilistic attractor} of $W$ in $(V, F)$, i.e. the set of vertices $v$ such that for every strategy profile $\bsigma$ in $\Game_{\|v}$ that uses only edges of $F$, there is a positive probability of reaching $W$.
    As a consequence of \cref{lm:secretlemma} (replacing the vertices of $W$ with terminal vertices), we have the following.
    
    \begin{claim}\label{claim:positiveattractorLinear}
      Given $W$, the set $\Attr(W, F)$ can be computed in time $\Oh(m)$.
    \end{claim}


Similarly, \cref{lm:secretlemma} enables us to compute the \emph{adversarial values} of each vertex, i.e., the best risk measure that the player controlling that vertex can ensure from that vertex when the other players are fully hostile.

    \begin{claim}\label{claim:adverserialXRLinear}
        For each $i$ and $v \in V_i$, the quantity:
        $$\val(v) = \inf_{\btau_{-i} \in \Strat_{-i}\Game_{\|v}} \sup_{\tau_i \in \Strat_i\Game_{\|v}} \X_i(\btau)$$
        can be computed in time $\Oh(m)$. 
    \end{claim}

Then, computing all those values can be done in time $\Oh(m^2)$.
We can therefore assume that those quantities $\val(v)$ are given with the input.

    \paragraph*{Cycle-friendly and cycle-averse cases}

We differentiate the two types of instances. 
    If there exists a player $i$ such that we have $y_i < 0$, then the requirement $\bx \leq \X(\bsigma) \leq \by$ implies that $\bsigma$ must almost surely reach a terminal vertex: we call that case the \emph{cycle-averse} case.
    If there is no such player, we are in the \emph{cycle-friendly} case.
    Our algorithm will work slightly differently in those two cases.
    However, the fundamental idea is still the same in both cases: we prune iteratively the set of edges, and each of the subsets $F \subseteq E$ which we obtain will induce a strategy profile $\bsigma^F$, in which the profitable deviations will be detected and used to prune new edges.
    However, the definition of $\bsigma^F$ differs in the cycle-averse and the cycle-friendly case.

    \paragraph*{Algorithm in the cycle-friendly case.}
    In the cycle-friendly case, for a given set of edges $F$, the strategy profile $\bsigma^F$ in the game $\Game_{\|v_0}$, is defined as follows: from each non-stochastic vertex $v$, when $v$ is seen for the first time, the strategy profile randomises uniformly between all the edges $vw \in F$.
    Later, when $v$ is visited again, it always repeats the same choice.
    Equivalently, each player initially chooses, at random, a positional strategy, and then follows it.
    If some player $i$ deviates and takes an edge that they are not supposed to take (be it an edge that does not belong to $F$ or an outgoing edge of a vertex from which a different edge has already been taken), then all the players switch to the positional strategy profile $\btau^{\dag i}$, where $\btau^{\dag i}_{-i}$ minimises the best risk measure that player $i$ can get (a positional such strategy profile exists by \cref{lm:secretlemma}), and $\tau^{\dag i}_i$ is some positional strategy.

    Our algorithm in the cycle-friendly case is presented in~\cref{algo:cyclefriendly}.
    Each step $k$ consists of identifying a new set of vertices $V_\bad^k$ that must be avoided.
    At step $k=0$, it is the set of terminal vertices that give some player $i$ a payoff that is larger than $y_i$, which would then make them have an off-constraints risk measure.
    At step $k \geq 1$, it is the set of vertices $v$ whose adversarial value $\val(v)$ is greater than the risk $z_i^k = \X_i(\bsigma^{E_k})$, where $i$ is the player controlling $v$.
    In other words, the vertices from which that player can have a profitable deviation.
    Note that it that second case, the computation of $V_\bad^k$ requires the computation of $z_i^k$, which can be done in time $\Oh(m)$ by computing the set of terminals that are accessible from $v_0$ in $(V, E_k)$, and by deciding whether the probability of reaching no terminal is positive: that will be the case if and only if there exists a positional strategy profile that uses only edges of $E_k$ (and therefore that $\bsigma^{E_k}$ is following with positive probability) such that with positive probability no terminal vertex is reached, which can be decided in time $\Oh(m)$ using \cref{lm:secretlemma}.

    Then, the positive probabilistic attractor $A_k = \Attr(v_\bad^k, E_k)$ is computed.
    If $k \geq 1$ and $v_0 \in A_k$, i.e., if it is not possible to avoid reaching the set $V_\bad^k$, the answer $\No$ is returned.
    Otherwise, the set $E_{k+1}$ is defined from $E_k$ by removing all the edges that lead from a vertex that does not belong to $A_k$ to a vertex that does, thus making sure that $V_\bad^k$ will never be reached.
    The algorithm stops when there is no more edge to remove.
    Then, the algorithm answers $\Yes$ and outputs the set $E_k$, as a succinct representation of the strategy profile $\bsigma^{E_{k+1}}$, if we have $z_i^k \geq x_i$ for each $i$, and answers $\No$ otherwise.

            \begin{algorithm}
            \begin{algorithmic}\caption{Constrained existence problem with optimists in the cycle-friendly case}\label{algo:cyclefriendly}
                \Procedure{CycleFriendly}{$\Game, \Bar{x},\Bar{y}$}
                    \State $k \gets 0$
                    \State $E_k\gets E$
                    \State $V_\bad^k = \{t\in T \mid \mu_i(t)> y_i\}$ 
                    \State $A_k \gets \Attr(V^k_\bad, E_k)$
                    \If{$v_0 \in A_k$}
                        \Return{$\No$}
                    \Else
                        \State $E_{k+1} \gets E_k \setminus \{uv\in E_k\mid u\not\in A_k\text{ and }v\in A_k\}$
                    \While{$k=0$\text{ or }$E_{k+1}\neq E_k$}
                        \State $k\gets k+1$
                        \State Compute $z^k_i = \X_i(\bsigma^{E_k})$ for each $i \in \Pi$
                        \State $V^k_\bad \gets \{v \mid \val(v) > z_i^k \text{ for } i \in \Pi \text{ such that } v \in V_i\}$
                        \State $A_k \gets \Attr(V^k_\bad, E_k)$
                        \If{$v_0 \in A_k$}
                            \Return{$\No$}
                        \Else
                            \State $E_{k+1} \gets E_k \setminus \{uv\in E_k\mid u\notin A_k\text{ and }v\in A_k\}$
                        \EndIf
                    \EndWhile
                    \EndIf
                    \If{$z_i^k\geq x_i$ for all players $i$}
                        \Return $(\Yes, E_{k+1})$
                    \Else{\text{ }}\Return{$\No$}
                    \EndIf
                \EndProcedure
            \end{algorithmic}
        \end{algorithm}

     \subparagraph*{Correctness in the cycle-friendly case.}
To prove the correctness of \cref{algo:cyclefriendly}, we first need to prove that the edge removals are such that all vertices always keep at least one outgoing edge, and that the stochastic ones always keep all of them, so that the strategy $\bsigma^{E_k}$ is always properly defined.

\begin{invariant}
    At each step $k$, every vertex $v \not\in V_?$ is such that $E_k(v) \neq \emptyset$, and every vertex $v \in V_?$ is such that $E_k(v) = E(v)$.
\end{invariant}

The proof is left to the reader.
We also need termination.

\begin{claim}
    \cref{algo:cyclefriendly} terminates.
\end{claim}

\begin{claimproof}
    At each step $k \geq 1$, we either have that the algorithm terminates, or that an edge is removed.
    The sequence $E_1, E_2, \dots$ is therefore strictly decreasing (note that we might have $E_0 = E_1$), hence it cannot be infinite.
\end{claimproof}

Now, to prove correctness, we will first prove the following claim.

        \begin{claim} \label{claim:zik}
            For each player $i$ and each index $k$, every strategy profile $\bsigma'$ that uses only edges of $E_k$ is such that $\X_i(\bsigma') \leq z_i^k$.
        \end{claim}

        \begin{claimproof}
            This result is a consequence of the fact that every payoff vector that can be obtained with positive probability in the strategy profile $\bsigma'$ is obtained with positive probability in the strategy profile $\bsigma^{E_k}$.
            
            Indeed, consider some payoff vector $\bz$ that has a positive probability to be generated in $\bsigma'$.
            If $\bz$ is obtained by reaching a terminal vertex, then that terminal vertex is accessible from $v_0$ in the graph $(V, E_k)$, and it therefore has a positive probability to be reached in $\bsigma^{E_k}$.
            
            If $\bz = (0)_i$ is obtained by reaching no terminal vertex, then by \cref{lm:secretlemma}, there exists a positional strategy profile $\btau$ that uses only edges of $E_k$ such that with positive probability, no terminal vertex is reached.
            Then, when following the strategy profile $\bsigma^{E_k}$, there is a positive probability that the players actually follow $\btau$.
            And therefore, there is also a positive probability to get the payoff vector $\bz = (0)_i$ in the strategy profile $\bsigma^{E_k}$.

            Since all players are optimists, the claim follows.
        \end{claimproof}

Note that this claim implies that the sequence $(z_i^k)_k$, for each $i$, is nondecreasing.

We can now prove correctness.
To do so, we need to prove two propositions: the algorithm recognises only positive instances, and recognises all of them.

\begin{proposition}
    The algorithm recognises only positive instances.
\end{proposition}

\begin{claimproof}
        Let us assume that the algorithm answers $\Yes$ at step $k$: let us show that the strategy profile $\bsigma^{E_k}$ is an XRSE that satisfies the desired constraints.
        Note that the algorithm does never answer $\Yes$ at step $0$, hence we necessarily have $k \geq 1$.

        \subparagraph*{The strategy profile $\bsigma^{E_k}$ satisfies $\bx \leq \X({\bsigma^{E_k} }) \leq \by$.}
        
        The lower bound is immediate since the algorithm answers $\Yes$ at step $k$ only if the strategy profile $\bsigma^{E_k}$ satisfies that constraint.

            Regarding the upper bound, observe that the set $E_1$ has been defined so that the set $\Attr(V_{\frownie}^0, E)$, and therefore the set $V_{\frownie}^0$, is not accessible from $v_0$ in the graph $(V, E_1)$, and therefore not in the graph $(V, E_{k+1})$.
            Thus, it is almost sure in $\bsigma^{E_{k+1}}$ that no vertex of $V_{\frownie}^0$ will ever be reached.
            In other words, all terminals that have a positive probability of being reached give each player $i$ a lower payoff than $y_i$.
            Now, if there is a positive probability that the play never will reach a terminal, that also does not give any player $i$ such a payoff, since we are in the cycle-friendly case.

        \subparagraph*{The strategy profile $\sigma^{E_k}$ is an XRSE.}
        
        Let $i$ be a player, and let $\sigma'_i$ be a deviation of player $i$ from $\bsigma^{E_k}$.
        We can assume without loss of generality that $\sigma'_i$ is pure.        
        Let $z' = \X_i({\bsigma^{E_k}_{-i}, \sigma'_i})$ be the extreme risk measure obtained by the player $i$.
        We want to prove that the deviation $\sigma'_i$ is not profitable, that is, we have $z' \leq z_i^k$.

        If the deviation $\sigma'_i$ uses only the edges of $E_k$, then it cannot be profitable by \cref{claim:zik}.
        But if it does use more edges, let us show that it cannot be a profitable deviation either.
 
\begin{claim}
    If there is a history $hv$ compatible with $\bsigma^{E_k}$ such that $v\sigma'_i(hv) \not\in E_k$, then we have $\X_i(\bsigma^{E_k}_{\|hv}, \sigma'_{i\|hv}) \leq z_i^k$.
\end{claim}

\begin{claimproof}
    After such a history, the strategy profile $\bsigma^{E_k}_{-i}$ follows the positional strategy profile $\btau^{\dag i}_{-i}$.
    By the definition of that strategy profile, we have $\X_i(\bsigma^{E_k}_{\|hv}, \sigma'_{i\|hv}) \leq \val(v)$.
    On the other hand, the vertex $v$ is accessible from $v_0$ in $(V, E_k)$, since it is visited with a positive probability in $\bsigma^{E_k}$.
    Therefore, it does not belong to the set $A_k$, and in particular not to the set $V_\bad^k$, which means that we have $\val(v) \leq z_i^k$.
    Hence, the conclusion follows. 
\end{claimproof}

In the general case, the payoffs that player $i$ obtains with positive probability in the strategy profile $(\bsigma^{E_k}_{-i}, \sigma'_i)$ are obtained either by using only edges that belong to $E_k$, or by using an edge that does not. In both cases, we have shown that player $i$ cannot get a payoff greater than $z_i^k$, which proves that the strategy profile $\bsigma^{E_k}$ is an XRSE.
\cref{algo:cyclefriendly} answers $\Yes$ only on positive instances, and outputs in that case a succinct representation of an XRSE matching the constraints.
\end{claimproof}

It now remains to prove the converse.
 
            \begin{proposition}
                \cref{algo:cyclefriendly} recognises all positive instances. 
            \end{proposition}
    
            \begin{claimproof}
           Let us assume that we have a positive instance, i.e., that there exists an XRSE $\bsigma$ with $\bx \leq \X(\bsigma) \leq \by$.
        Let us show that the algorithm will answer $\Yes$.
        To do so, we first prove the following invariant: if an edge is removed at some step, then it is never taken by the XRSE $\bsigma$.

        \begin{invariant} \label{inv:edgesnotused}
            For each $k \geq 0$, every edge that has positive probability to be eventually taken in $\bsigma$ belongs to $E_k$.
        \end{invariant}

        \begin{claimproof}
        We prove the invariant by induction. 
        
        \subparagraph*{Base case.} The case $k=0$ is immediate, since we have $E_0 = E$.
        
        Further, at step $k=1$, if the strategy profile $\bsigma$ uses eventually, with positive probability, an edge that does not belong to $E_1$, then it goes with positive probability to a vertex $v \in \Attr(V^0_\bad, E_0)$.
        Then, with positive probability, a terminal vertex will be reached that gives to some player $i$ a payoff greater than $y_i$, which is impossible.
        Therefore, such an edge cannot be taken in $\bsigma$.

        \subparagraph*{Induction step.} Let us assume that the invariant is true until step $k \geq 1$, and let us show that it holds at step $k+1$.    
        Let $uv$ be an edge that is used with positive probability when following $\bsigma$, and let us assume toward contradiction that it does not belong to $E_{k+1}$.
        Since the invariant is true at each step until $k$, we can assume that $uv$ has been removed at step $k$, i.e., that we have $uv \in E_k \setminus E_{k+1}$.
        Then, we have $u \not\in A_k$ and $v \in A_k$.
        The strategy profile $\bsigma$ has therefore positive probability of visiting the set $A_k$, and therefore the set $V_\bad^k$.
        Then, from a vertex of $V_\bad^k$, i.e., a vertex $v$ with $\val(v) > z_i^k$, player $i$ can deviate and get a risk measure strictly better than $z_i^k$.
        But since the invariant is true at step $k$, the strategy profile $\bsigma$ uses only vertices of $E_k$, and therefore, by \cref{claim:zik}, we have $\X_i(\bsigma) \leq z_i^k$: player $i$ has a profitable deviation in $\bsigma$, which is impossible.
\end{claimproof}

        We are now able to conclude.
        The answer $\No$ can be given in the two following cases:
            \begin{itemize}
                \item \emph{If at step $k$, we have $v_0 \in \Attr(V^k_\bad, E_k)$.}
                Then, the strategy profile $\bsigma$ visits the set $V^k_\bad$ with positive probability.
                With the same arguments that were used in the proof of \cref{inv:edgesnotused}, that is not possible.
                
                \item \emph{If during step $k$, no edge is removed, but we have $z^k_i < x_i$ for some player $i$.}
                Since $\bsigma$ uses only edges of $E_k$ by \cref{inv:edgesnotused}, we can apply \cref{claim:zik}, and obtain $\X_i(\bsigma) \leq z^k_i$, and therefore $\X_i(\bsigma) < x_i$: that case is therefore also impossible by definition of $\bsigma$.
            \end{itemize}
        None of those cases is possible, hence our algorithm will eventually answer $\Yes$.
    \end{claimproof}

    \paragraph*{Algorithm in the cycle-averse case}

    The algorithm and the structure of the proof will be similar.
    However, we need some significant modifications, especially in the definition of the strategy profiles $\bsigma^F$.

    In the cycle-averse case, for a given set of edges $F$, the strategy profile $\bsigma^F$, in the game $\Game_{\|v_0}$, is defined as follows: from each vertex $v \not\in V_?$, it randomises uniformly between all the edges $vw \in F$. Contrary to the cycle-friendly case, the outcome of such a randomisation has no influence on what will happen if $v$ is seen again.
    If some player $i$ deviates and takes an edge that they are not supposed to take (an edge that does not belong to $E_k$, then), then all the players switch to the positional strategy profile $\btau^{\dag i}$, where $\btau^{\dag i}_{-i}$ minimises the best risk measure that player $i$ can get (a positional such strategy profile exists by \cref{lm:secretlemma}), and $\tau^{\dag i}_i$ is some positional strategy.
 
    Our algorithm in the cycle-averse case is presented in~\cref{algo:cycleaverse}.
    Again, each step $k$ identifies a new set of vertices that must be avoided.
    Their definition depends now on the parity of $k$.
    When $k$ is even, it is the same as in the cycle-friendly case: the set $V^k_\bad$ is the set of vertices $v$ such that $\val(v) > z_i^k$, where $i$ is the player controlling $v$, and $A_k$ is the positive probabilistic attractor of $V_\bad^k$.
    When $k$ is odd, we define directly $A_k$ as the set of vertices from which whatever the players play, there is a positive probability of reaching no terminal vertex.
    
    Again, the computation of $V_\bad^k$ for an even step $k \geq 2$ requires the computation of $z_i^k$, which can be done in time $\Oh(m)$ by computing the set of terminals that are accessible from $v_0$ in $(V, E_k)$, and by deciding whether the probability of reaching no terminal is positive: that will be the case, now, if and only if there exists a vertex from which no terminal vertex is accessible, which can also be decided in time $\Oh(m)$.
    As for odd steps, the computation of $A_k$ can also be done in $\Oh(m)$ using \cref{lm:secretlemma}.
    
    If $k \geq 1$ and $v_0 \in A_k$, i.e., if it is not possible to avoid reaching the set $V_\bad^k$, the answer $\No$ is returned.
    Otherwise, the set $E_{k+1}$ is defined from $E_k$ by removing all the edges that lead from a vertex that does not belong to $A_k$ to a vertex that does, thus making sure that $V_\bad^k$ will never be reached.

    The loop stops when there is no more edge to remove, i.e., when we get $E_{k+2} = E_k$.
    Then, the algorithm answers $\No$ if we have $z_i^k < x_i$ for some $i$.
    Otherwise, it  performs \emph{final refinements}, defined as follows: first, it defines $F_0 = E_k$.
    Then, once $F_\l$ is defined for some $\l$, it checks whether there exists an edge $uv$ that matches the following conditions in the graph $(V, E_\l)$:
        \begin{enumerate}
            \item\label{itm:cuttableedge} the vertex $u$ is not stochastic and has several outgoing edges;
            \item\label{itm:nolessterminals} all the terminal vertices accessible from $v$ are also accessible from $v_0$ without using $uv$;
            \item\label{itm:nocycle} at least one terminal vertex is accessible from $u$ without using $uv$.
        \end{enumerate}
    In the following, we will refer to those conditions as Conditions~\ref{itm:cuttableedge}, \ref{itm:nolessterminals}, and \ref{itm:nocycle}.
    If there exists such an edge, then we define $F_{\l+1} = F_\l \setminus \{uv\}$.
    If there is no such edge, the algorithm stops there, answers $\Yes$, and returns $F_\l$ as a succinct representation of $\bsigma^{F_\l}$.

            \begin{algorithm}[t]
        \begin{algorithmic}\caption{Constrained existence problem with optimists in the cycle-averse case}\label{algo:cycleaverse}
                \Procedure{CycleAverse}{$\Game, \Bar{x},\Bar{y}$}
                    \State $k \gets 0$
                    \State $E_k\gets E$
                    \State $V_\bad^k = \{t\in T \mid \mu_i(t)> y_i\}$ 
                    \State $A_k \gets \Attr(V^k_\bad, E_k)$
                    \If{$v_0 \in A_k$}
                        \Return{$\No$}
                    \Else
                        \State $E_{k+1} \gets E_k \setminus \{uv\in E_k\mid u\not\in A_k\text{ and }v\in A_k\}$
                    \While{$E_{k+2}\neq E_k$\text{ or }$k \leq 1$}
                        \State $k\gets k+1$
                        \If{$k$ is even}
                            \State Compute $z^k_i = \X_i(\bsigma^{E_k})$ for each $i \in \Pi$
                            \State $V^k_\bad \gets \{v \mid \val(v) > z_i^k \text{ for } i \in \Pi \text{ such that } v \in V_i\}$
                            \State $A_k \gets \Attr(V^k_\bad, E_k)$
                        \Else
                            \State $A^k_\bad \gets \{v \mid \forall \btau \in \Strat_\Pi\Game_{\|v}, \prob_\btau(\Occ \cap T = \emptyset) > 0\}$
                        \EndIf
                        \If{$v_0 \in A_k$}
                            \Return{$\No$}
                        \Else
                            \State $E_{k+1} \gets E_k \setminus \{uv\in E_k\mid u\notin A_k\text{ and }v\in A_k\}$
                        \EndIf
                    \EndWhile
                    \EndIf
                    \If{$z_i^k < x_i$ for some player $i$}
                        \Return{$\No$}
                    \Else
                        \State $\l \gets 0$
                        \State $F_\l \gets E_k$
                        \Comment{Final refinement steps}
                        \While{there exists $uv$ satisfying Conditions~\ref{itm:cuttableedge}, \ref{itm:nolessterminals}, and \ref{itm:nocycle}}
                            \State $\l \gets \l+1$
                            \State $F_{\l+1} \gets F_\l \setminus \{uv\}$
                        \EndWhile
                        \Return{$(\Yes, F_\l)$}
                    \EndIf
                \EndProcedure
            \end{algorithmic}
        \end{algorithm}

     \paragraph*{Correctness in the cycle-averse case.}
The fact that $\bsigma^{E_k}$ and $\bsigma^{F_\l}$ are always correctly defined can be proved with arguments similar as those that were used for the cycle-friendly case.
We now focus on correctness properly said.
We first need the following properties.

\begin{invariant}\label{inv:terminalsaccessible}
    For every even $k > 0$, and for every $\l$, the graph $(V, E_k)$, or $(V, F_\l)$, contains no vertex that is accessible from $v_0$ and from which no terminal vertex is accessible.
\end{invariant}

\begin{claimproof}
    In the graph $(V, E_k)$ (for $k>0$ even), no induction is required: the set $E_k$ has been obtained after an odd step, in which the set $A_{k-1}$ has been made inaccessible.
    Thus, if we have a vertex $v$ from which no terminal vertex is accessible, it means that in the graph $(V, E_{k-1})$, all paths from $v$ to a terminal vertex were traversing a vertex of $A_{k-1}$, which implies that $v$ itself belonged to $A_{k-1}$, and is therefore not accessible from $v_0$ in $(V, E_k)$.

    This also proves that the invariant is true during the final refinements at step $\l = 0$.
    If now we assume that it is true at some step $\l$, then Condition~\ref{itm:nocycle} guarantees that it remains true at step $\l+1$.
\end{claimproof}

\begin{invariant} \label{inv:finishingtouchesconstantpayoff}
    If the algorithm switches to final refinements after step $k$, then for each step $\l$ of final refinements and for each player $i$, we have $\X_i(\bsigma^{F_\l}) = z_i^k$.
\end{invariant}

\begin{claimproof}
    The invariant is immediate for $\l=0$, since we have $F_\l = E_k$.
    Then, if it is true at step $\l$, it remains true at step $\l+1$.
    Indeed, Condition~\ref{itm:nolessterminals} guarantees that the set of terminal vertices accessible from $v_0$ in $(V, F_\l)$ is the same as in $(V, F_{\l+1})$.
    In other words, the terminal vertices that are reached with positive probability in $\bsigma^{F_\l}$ and $\bsigma^{F_{\l+1}}$ are the same.
    Moreover, \cref{inv:terminalsaccessible} guarantees that it is almost sure that some terminal vertex will be reached, in $\bsigma^{F_\l}$ as well as in $\bsigma^{F_{\l+1}}$.
    Therefore, the set of payoff vectors that have positive probability to be obtained is the same in both strategy profiles, hence the risk measures are the same.
\end{claimproof}

We can now prove correctness.
To do so, we need to prove two propositions: the algorithm recognises only positive instances, and recognises all of them.

\begin{proposition}
    The algorithm recognises only positive instances.
\end{proposition}

\begin{claimproof}
        Let us assume that the algorithm answers $\Yes$ at step $\l$ of the final refinements, after having switched to the final refinements loop at step $k$: let us show that the strategy profile $\bsigma^{F_\l}$ is an XRSE that satisfies the desired constraints.
        Note that the algorithm does never answer $\Yes$ at step $0$, hence we necessarily have $k \geq 1$.

 \subparagraph*{The strategy profile $\bsigma^{F_\l}$ satisfies $\bx \leq \X({\bsigma^{F_\l} }) \leq \by$.}
        
        The algorithm switches to the final refinements at step $k$ only if the strategy profile $\bsigma^{F_\l}$ satisfies $\X(\bsigma^{E_k}) \geq \bx$.
        Then, by \cref{inv:finishingtouchesconstantpayoff}, we also have $\X(\bsigma^{F_\l}) \geq \bx$.

        As for the upper bound, observe that the set $E_1$ has been defined so that the set $\Attr(V_{\frownie}^0, E)$, and therefore the set $V_{\frownie}^0$, is not accessible from $v_0$ in the graph $(V, E_1)$, and therefore not in the graph $(V, F_\l)$ either.
        Thus, it is almost sure in $\bsigma^{F_\l}$ that no vertex of $V_{\frownie}^0$ will ever be reached.
        In other words, all terminals that have positive probability to be reached give to each player $i$ a payoff smaller than $y_i$.
        That is sufficient to prove the lower bound, because it is almost sure, when following $\bsigma^{F_\l}$, that some terminal vertex will eventually be reached, by \cref{inv:terminalsaccessible}.

        \subparagraph*{The strategy profile $\bsigma^{F_\l}$ is an XRSE.}
        
        Let $i$ be a player, and let $\sigma'_i$ be a deviation of player $i$ from $\bsigma^{F_\l}$.
        Let $z' = \X_i({\bsigma^{F_\l}_{-i}, \sigma'_i})$ be the extreme risk measure obtained by player $i$.
        We want to prove that the deviation $\sigma'_i$ is not profitable, i.e., that we have $z' \leq z_i^k$ (since we have $\X_i(\bsigma^{E_\l}) = z_i^k$ by \cref{inv:finishingtouchesconstantpayoff}).
        To do so, we first show that player $i$ cannot obtain a payoff better than $z_i^k$ after using an edge that does not belong to $F_\l$.

        We first show that if the deviation $\sigma'_i$ uses only edges of $F_\l$, then it cannot be profitable.

\begin{claim}\label{claim:finishingtouchescycleimpossible}
    If it is almost sure, when following $(\bsigma^{F_\l}_{-i}, \sigma'_i)$, that only edges of $F_\l$ will be used, then the deviation $\sigma'_i$ is not profitable.
\end{claim}

\begin{claimproof}
    First, let us note that as long as player $i$ uses only edges that belong to $F_\l$, the strategy profile $\bsigma^{F_\l}$ behaves in a stationary way, and we can therefore assume without loss of generality that $\sigma'_i$ is positional.

    The payoff $z'$ may be obtained by reaching a terminal vertex: in that case, that terminal vertex is accessible from $v_0$ in $(V, F_\l)$, and therefore also reached with positive probability when following the strategy profile $\bsigma^{F_\l}$, hence $z' \leq z_i^k$.

    Let us show that it cannot be obtained by reaching no terminal.
    We proceed by contradiction: if, in the strategy profile $(\bsigma^{F_\l}, \sigma'_i)$, there is a positive probability of reaching no terminal when following that strategy profile, then there is a vertex that has positive probability of being visited infinitely often.
    We can then define the set $W$ of such vertices, i.e., the set $W = \{v \in V \mid \prob_{\bsigma^{F_\l}, \sigma'_i}(v \in \Inf) > 0\}$.
    Thus, when the strategy profile $(\bsigma^{F_\l}, \sigma'_i)$ is followed from a vertex of $W$, it is almost sure that no terminal vertex is reached, and that the set $W$ will never be left.
    We can then choose $w \in W$ such that it has positive probability of being reached without visiting any other vertex of $W$ before, i.e., such that there exists a history $hw$ from $v_0$ with $\Occ(h) \cap W = \emptyset$ (note that $h$ can be empty).
    
    On the other hand, in the graph $(V, F_\l)$, there is at least one terminal vertex accessible from $w$: all vertices from which no terminal is accessible are made themselves inaccessible at odd steps, the switch to final refinements loop happens only if there is no more edge to remove in that perspective, and Condition~\ref{itm:nocycle} guarantees that the final refinements loop leave at least one terminal vertex accessible from every vertex accessible from $v_0$.

    From each terminal $t$ accessible from $w$, we pick a simple path $h_0^t \dots h_{q_t}^t$ from $h_0^t = w$ to $h_{q_t}^t = t$ in the graph $(V, F_\l)$.
    Those paths define a directed acyclic graph (DAG) $D = (V_D, E_D)$ rooted at $w$, where all non-terminal vertices have at least one outgoing edge, with $V_D \subseteq V$ and $E_D \subseteq F_\l$.
    Now, since $\sigma'_i$ guarantees that no terminal vertex will be reached, each branch $h^t$ of that DAG is such that there exists a (smallest) index $j$ with $h^t_j \in W \cap V_i$, and $\sigma'_i(h^t_j) \neq h^t_{j+1}$.
    It may be the case that $h^t_j \sigma'_i(h^t_j) \in E_D$, i.e., that from $h^t_j$, player $i$ proceeds to an undetectable deviation and takes another branch of the DAG.
    But that cannot be the case for all $t$, otherwise, there would be a branch $h^t$ that would be followed with positive probability when following $(\bsigma^{F_\l}_{-i}, \sigma'_i)$ from $w$, and therefore a terminal vertex $t$ that would be reached with positive probability, which contradicts the definition of $w$.

    There must therefore exist an edge $uv \in F_\l \setminus E_D$, with $u \in V_D \cap V_i \cap W$
    We will show that such an edge should have been removed during the final refinements loop.
    First, it immediately satisfies Condition~\ref{itm:cuttableedge}.
    Moreover, the vertex $u$ is necessarily on a branch $h^t$ of $D$ that leads to a terminal vertex $t$, hence it satisfies Condition~\ref{itm:nocycle}.
    Finally, since $v$ is accessible from $w$ in $(V, F_\l)$, the terminal vertices that are accessible from $v$ in $(V, F_\l)$ are all accessible from $w$ in that same graph, and therefore are accessible from $w$ in the DAG $D$.
    Since $w$ is accessible from $v_0$ without visiting any vertex of $W$, and it particular without visiting $u$, it means that the terminal vertices accessible from $v$ are also accessible from $v_0$ without using the edge $uv$.
    In other words, the edge $uv$ satisfies Condition~\ref{itm:nolessterminals}, and should have been removed during the final refinements.

    This case is therefore impossible: when the players follow the strategy profile $(\bsigma^{F_\l}, \sigma'_i)$, it is almost sure that some terminal vertex will be reached, and that concludes the proof.
\end{claimproof}
        
But now, if the strategy $\sigma'_i$ does use edges that do not belong to $F_\l$, let us show that it cannot be a profitable deviation either.

\begin{claim}
    If there is a history $hv$ compatible with $\bsigma^{F_\l}$ such that $v\sigma'_i(hv) \not\in F_\l$, then we have $\X_i(\bsigma^{F_\l}_{\|hv}, \sigma'_{i\|hv}) \leq z_i^k$.
\end{claim}

\begin{claimproof}
    After such a history, the strategy profile $\bsigma^{F_\l}_{-i}$ follows the positional strategy profile $\btau^{\dag i}_{-i}$.
    By definition of that strategy profile, we have $\X_i(\bsigma^{F_\l}_{\|hv}, \sigma'_{i\|hv}) \leq \val(v)$.
    On the other hand, the vertex $v$ is accessible from $v_0$ in $(V, F_\l)$, since it is visited with positive probability in $\bsigma^{F_\l}$.
    Therefore, it does not belong to the set $A_k$ (if $k$ is even) or $A_{k-1}$ (if $k$ is odd), and in particular not to the set $V_\bad^k$ or $V_\bad^{k-1}$, which means that we have $\val(v) \leq z_i^k$.
    Hence the conclusion.
\end{claimproof}

In the general case, the payoffs that player $i$ obtains with positive probability in the strategy profile $(\bsigma^{F_\l}_{-i}, \sigma'_i)$ are either obtained using only edges that belong to $F_\l$, or by using an edge that does not: in both cases, we have shown that player $i$ cannot get a payoff greater than $z_i^k$, which proves that the strategy profile $\bsigma^{F_\l}$ is an XRSE.
\cref{algo:cyclefriendly} answers $\Yes$ only on positive instances, and outputs in that case a succinct representation of an XRSE matching the constraints.
\end{claimproof}

We will now prove the converse.

            \begin{proposition}
                \cref{algo:cycleaverse} recognises all positive instances. 
            \end{proposition}
    
            \begin{claimproof}
           Let us assume that we have a positive instance, i.e., that there exists an XRSE $\bsigma$ with $\bx \leq \X(\bsigma) \leq \by$.
        Let us show that the algorithm will answer $\Yes$.
        To do so, we first prove the following invariant: if an edge is removed at some step \emph{before the final refinements}, then it is never taken by the XRSE $\bsigma$.

        \begin{invariant} \label{inv:edgesnotused_bis}
            For each $k \geq 0$, every edge that has positive probability to be eventually taken in $\bsigma$ belongs to $E_k$.
        \end{invariant}

        \begin{claimproof}
        We prove the invariant by induction. 
        
        \subparagraph*{Base case.} The case $k=0$ is immediate, since we have $E_0 = E$.
        
        Further, at step $k=1$, if the strategy profile $\bsigma$ uses eventually, with positive probability, an edge that does not belong to $E_1$, then it goes with positive probability to a vertex $v \in \Attr(V^0_\bad, E_0)$.
        Then, with positive probability, a terminal vertex will be reached that gives to some player $i$ a payoff greater than $y_i$, which is impossible.
        Therefore, such an edge cannot be taken in $\bsigma$.

        \subparagraph*{Induction step.} Let us assume that the invariant is true until step $k \geq 1$, and let us show that it holds at step $k+1$.    
        Let $uv$ be an edge that is used with positive probability when following $\bsigma$, and let us assume toward contradiction that it does not belong to $E_{k+1}$.
        Since the invariant is true at each step until $k$, we can assume that $uv$ has been removed at step $k$, i.e., that we have $uv \in E_k \setminus E_{k+1}$.
        Then, we have $u \not\in A_k$ and $v \in A_k$.
        The strategy profile $\bsigma$ has therefore positive probability of visiting the set $A_k$.
        
        We must now distinguish the cases where $k$ is even or odd.
        If $k$ is even, then the strategy profile $\bsigma$ has therefore positive probability of visiting a vertex $v \in V_\bad^k$.
        If player $i$ is the player controlling $v$, then that player has a deviation in which they get risk measure at least $\val(v) > z_i^k$.
        Let us now note that when the players follow the strategy profile $\bsigma$, it is almost sure that a terminal vertex will eventually be reached, and that all the terminal vertices that have positive probability of being reached also have positive probability of being reached in $\bsigma^{E_k}$, since the invariant is true at step $k$: therefore, we have $z_i^k \geq \X_i(\bsigma)$, and player $i$ has a profitable deviation after $v$, which contradicts the fact that $\bsigma$ is an XRSE.

        If $k$ is odd, then, by definition of $A_k$, when the strategy profile $\bsigma$ is followed, there is a positive probability of reaching no terminal.
        But that is impossible in the cycle-averse case.

        The invariant is therefore necessarily still true at step $k+1$.
\end{claimproof}

        We are now able to conclude.
        The answer $\No$ can be given in the two following cases:
            \begin{itemize}
                \item \emph{If at step $k$, we have $v_0 \in \Attr(V^k_\bad, E_k)$.}
                Then, the strategy profile $\bsigma$ visits the set $V^k_\bad$ with positive probability.
                With the same arguments that were used in the proof of \cref{inv:edgesnotused_bis}, that is not possible.
                
                \item \emph{If during steps $k-1$ and $k$, no edge is removed, but we have $z^k_i < x_i$ for some player $i$.}
                Since $\bsigma$ uses only edges of $E_k$ by \cref{inv:edgesnotused_bis} and reaches almost surely a terminal (since we are in the cycle-averse case), we have $\X_i(\bsigma) \leq z^k_i$, and therefore $\X_i(\bsigma) < x_i$: that case is therefore also impossible by definition of $\bsigma$.
            \end{itemize}
        None of those cases is possible, hence our algorithm will eventually answer $\Yes$.
    \end{claimproof}

\paragraph*{Complexities.}
We consider here the complexity of the two algorithms.
Since at least one edge is removed every two steps, there are $\Oh(m)$ steps.
In each of them, we need $\Oh(p)$ calls to simple algorithms: computation of $z_i^k$, of $V_\bad^k$, of $A_k$.
Hence, the complexity $\Oh(pm^2)$.

In the cycle-averse case, when an output is asked, we need to add the final refinements loop: which consist of $\Oh(m)$ additional steps in which we check, for each of the $\Oh(m)$ remaining edges, whether they satisfy Conditions~\ref{itm:cuttableedge},~\ref{itm:nolessterminals}, and~\ref{itm:nocycle} in the finishing touches), which can be done in time $\Oh(m)$.
Hence the complexity $\Oh(pm^2 + m^3)$.
\end{proof}

\subsection{Proof of \cref{lm:ptimelowerbound}}\label{app:ptimelowerbound}
\ptimelowerbound*
\begin{proof}[Proof of \cref{lm:ptimelowerbound}]
        Given a deterministic two-player (between players $\Circle$ and $\Square$) zero-sum reachability game $\Game_{\|v_0}$ with target set of vertices $T$, we construct a simple stochastic game (with no stochastic vertices) where there is an XRSE $\bsigma$ satisfying $\X_\circ(\bsigma) = 1$ and $\X_\Box(\bsigma) = -1$ if and only if player $\Circle$ wins the game.  

        The game is simply obtained by assigning rewards on the zero-sum two player game as follows: we make all nodes in the target set $T$ of the reachability game as a terminal node where player $\Circle$ gets reward $1$ and player $\Square$ the reward $-1$. Recall that if no terminal is reached, both players get reward $0$.  

        If $\Circle$ has a strategy to win the reachability game, it is easy to see that the same strategy for $\Circle$, along with any strategy for $\Square$, will be an XRSE in that new game, and that it satisfies the constraint.
        Similarly, if on the other hand, player $\Square$ has a strategy to avoid the states $T$, then no strategy of $\Circle$ that gives her payoff $+1$ and gives player $\Square$ the payoff $-1$ will be an equilibrium, since $\Square$ can always deviate to the winning strategy in the reachability game that offers him the better payoff of $0$.
\end{proof}

\end{document}